\documentclass[12pt]{amsart}
\usepackage{fullpage,amsthm}
\usepackage{amsmath,amsfonts,amssymb} 
\usepackage{graphicx}

\input xy
\xyoption{all}

\newtheorem{thm}{Theorem}
\newtheorem{corl}[thm]{Corollary}
\newtheorem{lma}[thm]{Lemma} 
\newtheorem{prop}[thm]{Proposition}

\newtheorem{rem}[thm]{Remark}

\DeclareMathOperator{\ad}{ad}

\def\bar{\overline}

\def\C{\mathbb{C}}

\def\cO{\mathcal{O}}

\def\dd{\mathrm{d}}

\def\e{{[1]}}

\def\g{\mathfrak{g}}
\def\gf{\mathrm{gf}}
\def\gh{\textup{gh}}
\def\H{\mathcal{H}}

\def\half{\tfrac{1}{2}}

\def\id{\mathrm{id}}

\def\nn{\nonumber}

\DeclareMathOperator{\ord}{ord}

\def\R{\mathbb{R}}
\def\ren{\mathrm{ren}}

\def\su{\mathfrak{su}}

\def\tad{\textup{tad}}
\DeclareMathOperator{\tr}{Tr}
\def\tilde{\widetilde}

\DeclareMathOperator{\dvol}{dvol}
\DeclareMathOperator{\Vol}{Vol}

\def\Z{\mathbb{Z}} 

\def\a{\mathfrak{a}}
\def\b{\mathfrak{b}}
\def\c{\mathfrak{c}}
\def\d{\mathfrak{d}}
\def\e{\mathfrak{e}}
\def\f{\mathfrak{f}}
\def\g{\mathfrak{g}}
\def\h{\mathfrak{h}}
\def\k{\mathfrak{k}}
\def\l{\mathfrak{l}}
\def\m{\mathfrak{m}}

\title{Renormalization of the asymptotically expanded Yang--Mills spectral action}
\author{Walter D. van Suijlekom}
\address{Institute for Mathematics, Astrophysics and Particle Physics,
Radboud University Nij-megen, Heyendaalseweg 135, 6525 AJ Nijmegen, The Netherlands}
\date{November 16, 2011}

\begin{document}
\begin{abstract}
We study renormalizability aspects of the spectral action for the Yang--Mills system on a flat 4-dimensional background manifold, focusing on its asymptotic expansion. 
Interpreting the latter as a higher-derivative gauge theory, a power-counting argument shows that it is superrenormalizable. We determine the counterterms at one-loop using zeta function regularization in a background field gauge and establish their gauge invariance. Consequently, the corresponding field theory can be renormalized by a simple shift of the spectral function appearing in the spectral action. 

This manuscript provides more details than the shorter companion paper, where we have used a (formal) quantum action principle to arrive at gauge invariance of the counterterms. Here, we give in addition an explicit expression for the gauge propagator and compare to recent results in the literature.
\end{abstract}

\bibliographystyle{plainmath}
\maketitle

\section{Introduction}
Noncommutative geometry \cite{C94} has shown to be capable of describing Yang--Mills theories on the classical level, which further extends to the full Standard Model of high-energy physics \cite{CCM07}. One applies a so-called spectral action principle \cite{CC96,CC97} to a certain noncommutative manifold to arrive at a physical Lagrangian. In the low-energy limit, one recognizes the Lagrangian of the Standard Model, which can be perturbatively quantized using the usual physics textbook methods. It needs no stressing that this situation should be improved towards having a more intrinsic noncommutative geometrical description of the corresponding quantum theory. 

Recently, we have made some progress in this direction by showing that the asymptotically expanded spectral action for the Yang--Mills system -- interpreted as a higher-derivative field theory \cite{Sla71,Sla72b} -- is superrenormalizable \cite{Sui11}. The present paper gives full details of this result by presenting explicit formula for the gauge propagator and more importantly, we determine the form of the divergent part of the one-loop effective action. In {\it loc.cit.} this was derived from the formal quantum action principle (BRST-invariance of the one-loop effective action).
Since this is the only counterterm needed, and is proportional to the Yang--Mills action, this establishes renormalizability of the corresponding gauge field theory. 

\bigskip

Let us give an overview of the approach and results in this paper, whilst clarifying how our previous results are {\it not} in contradiction with \cite{ILV11}, as opposed to the authors' claim. We start with the full asymptotic expansion of the spectral action of Chamseddine and Connes \cite{CC96,CC97} in the case of the Yang--Mills system on a flat background manifold. That is, we study the asymptotics (as $\Lambda \to \infty$) of the spectral action:
\begin{equation}
\label{eq:sa}
S[A] := \tr f((D+A)/\Lambda)- \tr f(D/\Lambda)
\end{equation}
with $A= i \gamma^\mu A_\mu$ a Yang--Mills gauge field, minimally coupled to the Dirac operator $D$ on the flat background. We have subtracted the purely `gravitational' part $\tr f(D/\Lambda)$, being interested mostly in the gauge part of the spectral action. This also justifies our choice of a flat background manifold. 

The trace is over the $L^2$-spinor space and $f$ is a suitable function on $\R$. One concludes from this that the large eigenvalues of $D+A$ do not contribute to the spectral action, since $f(x) \to 0$ as $x \to \infty$ in order to make the trace well-defined. We note, however, that this does not mean that the corresponding field theory for $A$ is finite. Indeed, the large eigenvalues of $D+A$ have little to do with the high-frequency modes of $A$ itself. 

Chamseddine and Connes established in \cite{CC96} that the above spectral action \eqref{eq:sa} is given, asymptotically as $\Lambda \to \infty$, by
$$
S[A] \sim 
- \frac{f(0)}{24 \pi^2} \int_M \tr_N F_{\mu \nu} F^{\mu\nu} + \cO(\Lambda^{-1})  .
$$
Thus, as $\Lambda \to \infty$ the spectral action reduces to the Yang--Mills action
. In the original \cite{CC96,CC97} the authors adopted the Wilsonian viewpoint in which $\Lambda$ sets a physical energy scale. Instead, in the present paper we interpret $\Lambda$ as a regularizing cutoff parameter. 
This motivates the fact that we start with the asymptotic expansion rather than the full expression \eqref{eq:sa} as \cite{ILV11} do. In order to fully clarify the apparent mismatch with {\it loc.~cit.}, we make a slight change of notation with respect to our previous \cite{Sui11}: we will write $S^\Lambda[A]$ for the asymptotic expansion (as $\Lambda \to \infty$) of the spectral action $S[A]$ defined in Eq. \eqref{eq:sa}, stressing the role of $\Lambda$ as a regulator. Henceforth, we will refer to $S^\Lambda$ as the asymptotically expanded spectral action.

As said, we also take the terms proportional to $\Lambda^{-k}$ ($k >0$) into account. In fact, we will explicitly determine the tadpole terms and the free part of $S^\Lambda[A]$ to any order in $\Lambda$. In other words, we compute the parts linear and quadratic in $A$, respectively, exploiting a formula that we derived in \cite{Sui11}. This results in:
\begin{align*}
S^{\Lambda}_\tad[A] &= \frac{f_2}{4 \pi^2} \int \tr \partial_\mu A^\mu \\
S^\Lambda_0[A] &= - \int( \partial_\mu A_\nu  - \partial_\nu A_\mu)\varphi_\Lambda(\Delta) (\partial_\mu A_\nu - \partial_\nu A_\mu),
\end{align*}
where $\Delta = D^2$ is the Laplacian on the flat background, and $\varphi_\Lambda$ is the expansion
$$
\varphi_\Lambda(\Delta) := 
\sum_{k \geq 0} (-1)^k \Lambda^{-2k} f_{-2k} c_k \Delta^k
$$
with the $c_k$ certain combinatorial expressions, and the $f_i$ related to the derivatives of $f$ at zero. Note that for a simple non-abelian gauge group, the tadpole term vanishes, which fits nicely with \cite{CC06}.

After a suitable gauge-fixing, we explicitly determine the gauge propagator to any order in the derivatives. It is given by
$$
D_{\mu\nu}^{ab}(p; \Lambda) = \left[ g_{\mu\nu} - (1-\xi) \frac{p_\mu p_\nu}{p^2}\right] \frac{\delta^{ab}}{p^2 \varphi_\Lambda(p^2)} .
$$
Under suitable assumptions on the function $f$, this factor $\varphi_\Lambda(p^2)$ is a strictly positive polynomial, and improves the UV-behaviour of the corresponding perturbative quantum theory, which is typical for higher-derivative gauge theories as introduced in \cite{Sla71,Sla72b} (cf. \cite[Section 4.4]{FS80}). Consequently, after also incorporating the Faddeev--Popov ghost fields, we show in Section \ref{sect:ren-sa} that the asymptotically expanded spectral action $S^\Lambda[A]$ for the Yang--Mills system is power-counting superrenormalizable. We use zeta function regularization in a background field gauge -- exploiting the explicit forms for the heat invariants for higher-order Laplacians derived by \cite{Gil80} -- to determine in Section \ref{sect:1loop-sa} the form of the counterterm, which is proportional 
to the Yang--Mills action.
They can thus safely be subtracted from the spectral action, involving only the Yang--Mills term in $S^\Lambda[A]$. 
We conclude that $S^\Lambda[A]$ can be renormalized through a redefinition of the coefficient $f(0)$
:
\begin{gather*}
f(0) \mapsto f(0) +24 \pi^2 \left( c  + \tilde c\right)
\left(\frac{1}{z} + 2k \ln \mu\right)
\end{gather*}
where $\mu$ is a mass scale. Furthermore, $c$ and $\tilde c$ are constants 
determined from the computations of the divergent part. This shows renormalizability of the asymptotic expansion $S^\Lambda[A]$ for the Yang--Mills system. 

We explicitly compute the counterterms for the spectral action up to 6'th order in $A$. Even though this theory is not superrenormalizable (which would require also terms of 8'th order in $A$), it serves as an illustration of our methods.
Finally, we comment on the subtle relation between the renormalized asymptotically expanded spectral action for the Yang--Mills system and the usual renormalization of Yang--Mills theory.

\bigskip

Let us end this section by clarifying the apparent mismatch of the above results with \cite{ILV11} (see also Remark \ref{rem:ILV} below). There it was shown that the quadratic part of the spectral action decays as $1/p^4$ as $p \to \infty$, which seems to contradict the above results. However, note that the field theories that one is comparing are different. 
In the above, and also in \cite{Sui11}, we adopt an asymptotic expansion of the spectral action in the parameter $\Lambda$. Together with a suitable choice of the function $f$, this allows for the corresponding action functional $S^\Lambda[A]$ to have only finitely many terms when expanded in $\Lambda$, thus defining a local field theory. 
In contrast, the authors in \cite{ILV11} consider the spectral action $S[A]$ as defined in Eq. \eqref{eq:sa} (without expanding in $\Lambda$). This defines a different -- in fact non-local -- field theory, with correspondingly different large momentum behaviour.


\section{The Yang--Mills system}

The object of study in this paper is the {\it spectral action} for the Yang--Mills (YM) system on a compact background manifold. It is given by the relatively simple formula:
$$
S[A] := \tr f(D_A/\Lambda) - \tr f(D/\Lambda).
$$
This spectral action has firm roots in the noncommutative geometrical description of the Yang--Mills system. One considers a Dirac operator with coefficients in a $SU(N)$-vector bundle equipped with a connection $A$. That is, locally we have 
$$
D_A = i \gamma^\mu ( \nabla_\mu +  A_\mu)
$$
with $\nabla_\mu$ the spin connection on a Riemannian spin manifold $M$ and $A_\mu$ a skew-hermitian traceless matrix. The (hermitian) Dirac gamma matrices satisfy $\{ \gamma_\mu, \gamma_\nu \} = 2 g_{\mu\nu}$ and are represented on spinor space $S_x$ for each $x \in M$. The Dirac operator then acts as a self-adjoint operator in the Hilbert space $\H$ of $M_N(\C)$-valued spinors:
$$
\H := L^2(M,S) \otimes M_N(\C).
$$
Such a construction is {\it noncommutative} in the sense that now coordinates on $M$ are naturally $M_N(\C)$-valued as well, which leads to consider the basic set of data 
$$
( C^\infty(M)\otimes M_N(\C), L^2(M,S) \otimes M_N(\C), D_A \otimes 1).
$$
This is the first example of a {\it spectral triple}. We will not go into further details on this, but refer to \cite{C94,GVF01,CM07} for more details.
For simplicity, we take $M$ to be flat ({\it i.e.} vanishing Riemann curvature tensor) and 4-dimensional. 
Furthermore, we will assume that $f$ is a Laplace--Stieltjes transform:
$$
f(x) = \int_{t>0} e^{-tx^2} d\mu(t).
$$

\begin{prop}[\cite{CC97}]
In the above notation, there is an asymptotic expansion (as $\Lambda \to \infty$):
\begin{equation}
\label{sa-eym}
S[A]  \sim \sum_{m \geq 0} \Lambda^{4-m} f_{4-m} \int_M a_m (x,D_A^2),
\end{equation}
in terms of the Seeley--De Witt invariants of $D_A^2$.
The coefficients are defined by $f_k := \int t^{-k/2} d\mu(t)$; in particular $f_0 = f(0)$.
\end{prop}
We will denote the right-hand side of Equation \eqref{sa-eym} by $S^\Lambda[A]$.
Recall that the Seeley--De Witt coefficients $a_m(x,D_A^2)$ are gauge invariant (and coordinate independent) polynomials in the fields $A_\mu$. Indeed, the Weitzenb\"ock formula gives
\begin{equation}
\label{eq:weitzenbock}
D_A^2 =- (\partial_\mu +A_\mu) (\partial^\mu + A^\mu)- \frac{1}{2} \gamma^\mu \gamma^\nu F_{\mu\nu}
\end{equation}
in terms of the curvature $F_{\mu\nu} = \partial_\mu A_\nu - \partial_\nu A_\mu + [A_\mu,A_\nu]$ of $A_\mu$. Consequently, a Theorem by Gilkey \cite[Theorem 4.8.16]{Gil84} shows that (in this case) $a_m$ are polynomial gauge invariants in $F_{\mu\nu}$ and its covariant derivatives. The {\it order} $\ord$ of $a_m$ is $m$, where we set on generators:
$$
\ord A_{\mu_1; \mu_2\cdots \mu_k} = k.
$$
Consequently, the curvature $F_{\mu\nu}$ has order $2$, and $F_{\mu_1 \mu_2; \mu_3 \cdots \mu_k}$ has order $k$. For example, $a_4(x,D_A^2)$ is proportional to $\tr F_{\mu\nu}F^{\mu\nu}$. In fact, we have:
\begin{thm}[Chamseddine-Connes \cite{CC97}]
\label{thm:CC}
The spectral action for the above Yang--Mills system is given, asymptotically as $\Lambda \to \infty$, by
$$
S[A] \sim \frac{f_4 N^2 \Lambda^4}{2 \pi^2} \Vol(M)- \frac{f_0}{24 \pi^2} \int_M \tr_N F_{\mu \nu} F^{\mu\nu} + \cO(\Lambda^{-1})  
$$
where $\tr_N$ denotes the trace in (the adjoint representation) $\su(N)$.
\end{thm}
This appearance of the Yang--Mills action at lowest order is the main motivation to study this model. We aim at a better understanding also of the terms in $S^\Lambda[A]$ proportional $\Lambda^{-k}$ ($k >0$). First, we compute the coefficients $f_k$ explicitly. 
\begin{lma}
\label{lma:coef}
The constants $f_{k} := \int t^{-k/2} d\mu(t)$ $(k \in \Z)$ are given by
\begin{enumerate}
\item $k> 0$: $f_k = \frac{2}{\Gamma\left(\frac{k}{2}\right)} M_{k-1}[f]$ with $M_{k-1}$ the $k-1$'th moment of $f$ ,
\item $k\geq 0$: $f_{-2k} = \frac{(-1)^k f^{(2k)}(0)}{(2k-1)!!}.
$
\end{enumerate}
\end{lma}
\proof
(1) was already derived in a slightly different form in \cite{CC97} (cf. \cite[Sect. 1.11]{CM07}). In our notation, we substitute $t^{-k/2}$ in the definition of $f_k$ using the Mellin transform (cf. Eq. \eqref{eq:mellin} below):
$$
f_k = 
\frac{1}{\Gamma\left(\frac{k}{2}\right)}\int_{t>0} \int_{v>0} e^{-tv} v^{k/2-1} d\mu(t)dv =  \frac{1}{\Gamma\left(\frac{k}{2}\right)} \int_{v>0} v^{k-1} f(v) dv = \frac{2}{\Gamma\left(\frac{k}{2}\right)} M_{k-1}.
$$
(2) We derive for the even derivatives of $f$:
$$
f^{(2k)}(x) = \int_{t>0} e^{-tx^2/2} H_{2k}(\sqrt t x) t^{k} d\mu(t)
$$
in terms of the Hermite polynomials $H_{n}(x) \equiv (-1)^n e^{x^2/2} (d/dx)^n e^{-x^2/2}$. Evaluating both sides at zero gives the desired result, using in addition that $H_{2k}(0)= (-1)^k (2k-1)!!$.
\endproof

Next, we take a closer look at the terms of lowest order in $A$ in $S[A]$. In particular, we will derive formulas for the {\it tadpole term}
$$
S_{\textup{tad}}[A] = \frac{\dd}{\dd u} S[uA] \bigg|_{u=0},
$$
and the {\it free action}
\begin{equation}
\label{eq:free-sa}
S_0[A] = \frac{1}{2} \frac{\dd}{\dd u}  \frac{\dd}{\dd v} S[uA+vA]  \bigg|_{u=v=0}.
\end{equation}
In the next subsection, we will expand the latter asymptotically as $\Lambda \to \infty$, in terms of the above coefficients $f_{k}$. 
For a fully rigorous derivation of the above formulas, and formulas for the higher order terms using G\^ateaux derivatives in a more general functional analytical setting, we refer to \cite{Sui11}. Here, we only sketch the derivation which is based on the following result:
\begin{lma}
\label{lma:heatoperator}
Let $P(A)= DA+AD+A^2$ with $A= i \gamma^\mu A_\mu$. Then
$$
e^{-t (D_A)^2} = e^{-tD^2} - t \int_0^1 ds ~ e^{-st (D_A)^2} P(A) e^{-(1-s)t D^2}.
$$
\end{lma}
\begin{proof}
Note that $e^{-t D_A^2}$ is the unique solution of the Cauchy problem
$$
\left\{ \begin{array}{r} \left( d_t + D_A^2 \right) u(t) = 0 \\ u(0)=1 \end{array} \right.
$$
with $d_t = d/dt$. Using the fundamental theorem of calculus, we find that 
\begin{align*}
d_t \left[ e^{-tD^2} - \int_0^t dt' e^{-(t-t')D_A^2} P(A) e^{-t'D^2} \right]
&= -D^2 e^{-tD^2} - P(A) e^{-t D^2} \\
& \qquad \qquad +  \int_0^t dt' D_A^2 e^{-(t-t')D_A^2} P(A) e^{-t'D^2} \\
&= -D_A^2 \left( e^{-tD^2} -  \int_0^t dt' e^{-(t-t')D_A^2} P(A) e^{-t'D^2}  \right)
\end{align*}
showing that the bounded operator $e^{-tD^2} - \int_0^t dt' e^{-(t-t')D_A^2} P(A) e^{-t'D^2}$ also solves the above Cauchy problem. 
\end{proof}

\begin{prop}
\label{prop:firstder}
The tadpole term is given by
$$
S_\tad[A] =  - \int_{t>0}  t \Lambda^{-2} \tr_\H \left[ \{D,A\} e^{-t (D/\Lambda)^2} \right]d\mu(t) = \Lambda^{-2y} \tr_\H A f'(D/\Lambda).
$$
\end{prop}
\begin{proof}
Using Lemma \ref{lma:heatoperator} we obtain after substituting $t \mapsto t/\Lambda^2$:
\begin{align*}
\frac{1}{u}\left(S[uA] - S[0] \right)
&= \frac{1}{u} \int  \tr \left( e^{-t (D_{uA}/\Lambda)^2} - e^{-t(D/\Lambda)^2} \right)d\mu(t)\\
&= -\frac{1}{u} \int  t \Lambda^{-2} \tr \int_0^1 e^{-st (D_{uA}/\Lambda)^2} P(uA) e^{-(1-s)t (D/\Lambda)^2} ds d\mu(t)\\
& \to -\Lambda^{-2} \int t \tr (DA+AD) e^{-t (D/\Lambda)^2}d \mu(t) \qquad \text{ as } u \to 0.
\end{align*}
Finally, the identity $f'(x) = \int (-2t x) e^{-t x^2 } d\mu(t) $ yields the displayed formula. 
\end{proof}

\begin{prop}
\label{prop:secondder}
The free part of the spectral action is given by
\begin{multline*}
S_0[A] = \int_{t>0}  \bigg\{ -t \Lambda^{-2}   \tr \left[ A^2 e^{-t (D/\Lambda)^2} \right] \\
+ \frac{1}{2} t^2 \Lambda^{-4}   \tr \int_0^1 \{ D, A\} e^{-st (D/\Lambda)^2} \{D,A\} e^{-(1-s)t (D/\Lambda)^2} ds \bigg\} d\mu(t).
\end{multline*}
\end{prop}
\begin{proof}
As in the proof of Proposition \ref{prop:firstder} we derive
$$
\frac{\dd}{\dd v} S[uA +vA] \bigg|_{v=0} = - \int t  \Lambda^{-2}  \{ D_{uA},A \} e^{-t (D_{uA}/\Lambda)^2} d\mu(t).
$$
Applying Lemma \ref{lma:heatoperator} once more, we find
\begin{multline*}
\frac{\dd}{\dd u}\frac{\dd}{\dd v} S[uA +vA] \bigg|_{u=v=0} \!\!\!\!\!= -2 \int t \Lambda^{-2}   \tr A^2 e^{-t (D/\Lambda)^2} d\mu(t)\\ + \int t^2 \Lambda^{-4}   \tr \{ D,A\} \int_0^1 e^{-st(D/\Lambda)^2} \{D,A\} e^{-(1-s)t (D/\Lambda)^2} ds d\mu(t),
\end{multline*}
as claimed.
\end{proof}
Note that the above formulas for $S_\tad$ and $S_0$ are exact, and not asymptotic expansions; below we will consider the asymptotics of $S_0[A]$ for large $\Lambda$.

\subsection{Local expressions for the tadpole and free part}
The aim of this section is to derive local expressions for the above tadpole and free part of the spectral action for the Yang--Mills system. For the latter, this is possible by adopting an asymptotic expansion. We will use heat kernel techniques which we briefly recall, referring for more details to \cite{BGV92}. 

First, the Dirac operator can be related to the Laplacian $\Delta$ on $M$ via Weitzenb\"ock's formula:
$$
D^2 = \Delta := -g^{\mu\nu} \partial_\mu \partial_\nu.
$$
The heat kernel for $\Delta$ is simply given by 
\begin{equation}
\label{heat-kernel}
k_t(x,y) = (4 \pi t)^{-2} e^{-\| x-y\|/4t}.
\end{equation}
It satisfies 
\begin{equation}
\label{heat-kernel-expansion}
\int k_t(x,y)\psi(y)dy \sim  \sum_{k=0}^\infty \frac{(-t)^k}{k!}( \Delta^k \psi)(x) 
\end{equation}
asymptotically as $t \to 0$. This reflects the fact that $k_t$ is the kernel of the heat operator in the sense that
$$
e^{-t \Delta} \psi(x) = \int_M k_t(x,y) \psi(y) dy; \qquad (\psi \in L^2(M)).
$$
The following relations will be convenient later:
\begin{equation}
\label{heat-kernel-relations}
\begin{aligned}
\partial_\mu^x k_t(x,y) &= -  \frac{x_\mu - y_\mu}{2t} k_t(x,y) = -\partial_\mu^y k_t(x,y),\\
\partial_\mu^x \partial_\nu^y k_t(x,y) &=   \frac{g_{\mu\nu}}{2t}  k_t(x,y) + \frac{(x_\mu - y_\mu)(y_\nu - x_\nu)}{4 t^2} k_t(x,y),\\
\partial_\mu^x \partial_\nu^x k_t(x,y) &= -  \frac{g_{\mu\nu}}{2t}  k_t(x,y) + \frac{(x_\mu - y_\mu)(x_\nu - y_\nu)}{4 t^2} k_t(x,y). 
\end{aligned}
\end{equation}

\begin{thm}
The tadpole term for the Yang--Mills system is given by
$$
S_\tad[A] = \frac{f_2 \Lambda^2}{4 \pi^2} \int_M \tr_N \partial_\mu A^\mu
$$
which vanishes for $A_\mu$ a $\su(N)$-gauge field.
\end{thm}
\proof
The kernel of the operator $\{D,A\} e^{-t (D/\Lambda)^2}$ appearing in Proposition \ref{prop:firstder} is given by
$$
- \gamma^\mu \gamma^\nu \left( \partial_\mu A_\nu(x) + A_\nu \partial_\mu^x+ A_\mu \partial^x_\nu \right) k_{t/\Lambda^2}(x,y).
$$
Taking the trace corresponds to integrating this heat kernel over the diagonal (and taking the trace over Dirac matrices and $M_N(\C)$), so that with Eq. \eqref{heat-kernel} and \eqref{heat-kernel-relations}  we find
$$
S_\tad [A] =  \int_{t>0} \frac{\Lambda^2}{4 \pi^2 t} d\mu(t) \int_M \tr_N \partial_\mu A^\mu 
$$
using $\tr \gamma^\mu \gamma^\nu = 4 g^{\mu\nu}$. 
\endproof

\begin{thm}
\label{thm:free-sa}
There is the following asymptotic expansion (as $\Lambda \to \infty$) for the free part of the spectral action on a flat background manifold $M$
$$
S_{0}[A] \sim S^\Lambda_{0}[A] :=- \sum_{k \geq0} (-1)^k c_k f_{-2k} \Lambda^{-2k} \int \tr_N \hat F^{\mu\nu} \Delta^k ( \hat F_{\mu\nu})
$$
where $\Delta$ is the Laplacian on $(M,g)$, $\hat F_{\mu\nu} = \partial_\mu A_\nu - \partial_\nu A_\mu$ and $c_k$ are the following positive constants:
$$
c_k =  \frac{1}{8 \pi^2} \frac{(k+1)!}{(2k+3)(2k+1)!}.
$$
\end{thm}

\proof
We consider the first term in the expression for $S_0[A]$ derived in Proposition \ref{prop:secondder}. After writing $A= i \gamma^\mu A_\mu$, using the explicit form of the heat kernel and the property that $\tr \gamma^\mu \gamma^\nu =4 g^{\mu\nu}$, we find:
\begin{equation}
\label{eq:mass-term}
- \Lambda^{-2} \int_{t>0} t   \tr  A^2 e^{-t D^2/\Lambda^2} d\mu(t) =  \frac{4 f_2 \Lambda^2}{(4 \pi )^2}  \int_M \tr_N A_\mu A^\mu.
\end{equation}
The second expression in $S_0[A]$ is more involved, we first determine (suppressing the $\Lambda$-dependence until we have finished the proof of Lemma \ref{lem:kernel-term}) for $\psi \in \H$:
\begin{align}\nn
\{ D,A\} e^{-st D^2} \{ D,A\} e^{-(1-s)tD^2} \psi(x) &= 
 \int dy dz 
\left( -\partial_\mu A^\mu(x) - 2 A_\mu \partial_x^\mu -\tfrac{1}{2} \gamma^\mu \gamma^\nu \hat F_{\mu\nu}(x) \right) 
k_{st}(x,y) \\ &  \times \left(- \partial_\rho A^\rho(y) - 2  A_\rho \partial_y^\rho-\tfrac{1}{2} \gamma^\rho \gamma^\sigma \hat F_{\rho\sigma}(y) \right) k_{(1-s)t}(y,z) \psi(z)
\label{eq:2ndterm}
\end{align}
with $\hat F_{\mu\nu} = \partial_\mu A_\nu - \partial_\nu A_\mu$. Indeed, this follows by substituting
$$
\{ D, A\} = - \partial_\mu A^\mu - 2  A_\mu \partial^\mu-\frac{1}{2} \gamma^\mu \gamma^\nu \hat F_{\mu\nu}
$$
as in Weitzenb\"ocks formula \eqref{eq:weitzenbock}.
Note that $\id$ and $\gamma^\mu \gamma^\nu$ $( \mu \neq \nu)$ are orthogonal with respect to the Hilbert--Schmidt inner product. We derive the local form of the resulting expressions in a series of Lemma's. 
\begin{lma}
\label{lem:F-term}
\begin{multline*}
\int_0^1 ds \int_{M \times M} \tr \half \gamma^\mu \gamma^\nu \hat F_{\mu\nu}(x)\half \gamma^\rho \gamma^\sigma \hat F_{\rho\sigma}(y) k_{st}(x,y) k_{(1-s)t}(y,x) \\
\sim - 2 (4 \pi t)^{-2} \sum_{k\geq 0}   \int \tr_N \hat F_{\mu\nu}\frac{k!}{(2k+1)!} (-t \Delta)^k \hat F^{\mu\nu}.
\end{multline*}
\end{lma}
\begin{proof}
From the explicit form of $k_t(x,y)$ we derive that
\begin{equation}
\label{heat-kernel-product}
k_{st}(x,y) k_{(1-s)t}(y,x) = (4 \pi t)^{-2} k_{s(1-s)t}(x,y).
\end{equation}
The result then follows from the asymptotic expansion of $e^{-s(1-s)t \Delta}$, the standard integrals,
\begin{equation}
\label{simplex-integral}
\int_0^1 s^k(1-s)^l  = \frac{k! l!}{(k+l+1)!},
\end{equation}
and the trace formulas $\tr \gamma^\mu \gamma^\nu \gamma^\rho \gamma^\sigma = 4 \left( g^{\mu\nu} g^{\rho\sigma} - g^{\mu\rho}g^{\nu \sigma} + g^{\mu \sigma} g^{\nu \rho}\right)$ in spinor space.
\end{proof}
The remaining term from Eq. \eqref{eq:2ndterm} becomes after a series of integration by parts:
\begin{multline*}
\int_{M \times M}
\tr \left( \partial_\mu A^\mu(x) + 2 A_\mu \partial_x^\mu  \right) k_{st}(x,y) 
 \left(\partial_\nu A^\nu(y) + 2 A_\nu \partial_y^\nu  \right) k_{(1-s)t}(y,x)  \\ = \int_{M \times M}\tr  A^\mu(x) A^\nu(y) \left[  2 \partial_\mu^x k_{st}(x,y) \partial_\nu^y k_{(1-s)t} (x,y) -2 \partial_\mu^x \partial_\nu^y k_{st} (x,y) k_{(1-s)t}(x,y) \right].
\end{multline*}

\begin{lma}
\begin{multline*}
 2 \partial_\mu^x k_{st}(x,y) \partial_\nu^y k_{(1-s)t} (x,y) - 2 \partial_\mu^x \partial_\nu^y k_{st} (x,y) k_{(1-s)t}(x,y)  \\=  -2 (4 \pi t)^{-2} \left[ s(1-s) - (1-s)^2\right] \partial_\mu^y \partial_\nu^y k_{s(1-s)t}(x,y) -  2 (4 \pi t)^{-2} \frac{g_{\mu\nu}}{t} k_{s(1-s)t}(x,y). 
\end{multline*}
\end{lma}
\begin{proof}
Using the above relations \eqref{heat-kernel-relations} in combination with Equation \eqref{heat-kernel-product} we find for the first term:
$$
2\partial_\mu^x k_{st}(x,y) \partial_\nu^y k_{(1-s)t} (x,y)
=  -2(4 \pi t)^{-2} s(1-s) \partial_\mu^y \partial_\nu^y k_{s(1-s)t}(x,y) - (4\pi t)^{-2} \frac{g_{\mu\nu}}{t} k_{s(1-s)t}(x,y),
$$
and for the second
$$
- 2 \partial_\mu^x \partial_\nu^y k_{st} (x,y) k_{(1-s)t}(x,y) 
=   2(4\pi t)^{-2} (1-s)^2 \partial_\mu^y \partial_\nu^y k_{s(1-s)t}(x,y) - (4 \pi t)^{-2} \frac{g_{\mu\nu}}{t} k_{s(1-s)t}(x,y).
$$
\end{proof}
We now combine the above results:
\begin{lma}
\label{lem:kernel-term}
\begin{multline*}
\int_{M \times M}  \int_0^1 ds
\tr \left(  \partial_\mu A^\mu(x) + 2 A_\mu \partial_x^\mu  \right) k_{st}(x,y) 
 \left( \partial_\nu A^\nu(y) + 2 A_\nu \partial_y^\nu  \right) k_{(1-s)t}(y,x)  \\ \sim- 2( 4 \pi t)^{-2} t^{-1} \int_M  \tr A^\mu A_\mu 
-2 (4 \pi t)^{-2} \sum_{k=0}^\infty \frac{(k+1)!}{(2k+3)!}  \int_M \tr A^\mu \left(- g_{\mu\nu} \Delta- \partial_\mu \partial_\nu \right) (-t\Delta)^k A^\nu.
\end{multline*}
\end{lma}
\begin{proof}
First, we derive from the above Lemma:
\begin{multline*}
\int_{M \times M} \int_0^1 ds
\tr \left(  \partial_\mu A^\mu(x) + 2 A_\mu \partial_x^\mu  \right) k_{st}(x,y) 
 \left( \partial_\nu A^\nu(y) + 2 A_\nu \partial_y^\nu  \right) k_{(1-s)t}(y,x)  \\ \sim  -2 (4 \pi t)^{-2} \sum_{k \geq 0} \int_M \int_0^1 ds \tr A^\mu(x)  \frac{(-t)^k}{ k! t} \Delta^k A_\mu(x) s^k (1-s)^k
\\ - 2  (4 \pi t)^{-2} \sum_{k=0}^\infty \int_M \int_0^1 ds \tr A^\mu(x) \partial_\mu \partial_\nu \frac{(-t)^k s^k (1-s)^k }{k!} \Delta^k A^\nu(x) \left[ s(1-s)- (1-s)^2\right].
\end{multline*}
With Equation \eqref{simplex-integral} the first term equals
\begin{multline*}
2 (4 \pi t)^{-2} \sum_{k \geq 0} \int_M \tr A^\mu(x)  \frac{(-t)^{k-1} k!}{(2k+1)!} \Delta^k A_\mu(x) = -2 (4 \pi t)^{-2} t^{-1}  \int_M \tr A^\mu A_\mu  \\+  2 (4 \pi t)^{-2} \sum_{k \geq 0} \int_M \tr A^\mu(x)  \frac{(-t)^{k} (k+1)!}{(2k+3)!} \Delta^{k+1} A_\mu(x). 
\end{multline*}
On the other hand, the second term becomes with \eqref{simplex-integral}
$$
 2  (4 \pi t)^{-2} \sum_{k=0}^\infty \int_M \tr A^\mu(x) \partial_\mu \partial_\nu \frac{(-t)^k (k+1)!}{(2k+3)!} \Delta^k A^\nu(x).
$$
These last two formulas combine to give the desired result. 
\end{proof}

After integrating over $d \mu(t)$ and taking the trace over spinor and $\su(N)$-indices, this combines with Lemma \ref{lem:F-term} to yield the final result for the second term in Proposition \ref{prop:secondder}:
\begin{align*}
&\int_{t>0}  \frac{1}{2} t^2 \Lambda^{-4}   \tr \int_0^1 \{ D, A\} e^{-st (D/\Lambda)^2} \{D,A\} e^{-(1-s)t (D/\Lambda)^2} ds \bigg\} d\mu(t)\\
&\qquad = -\int_{t>0} 4 (4 \pi t)^{-2} t \Lambda^2 d\mu(t) \int_M \tr_N A^\mu A_\mu 
\\
& \qquad\qquad-\int_{t>0} 2(4 \pi t)^{-2} t^2 d\mu(t)  \sum_{k=0}^\infty c_k' \int_M \tr_N \hat F^{\mu\nu} (-t\Delta/\Lambda^2)^k \hat F_{\mu\nu} \\
&\qquad =- \frac{4 f_2 \Lambda^2 }{(4 \pi)^2} \int_M \tr_N A^\mu A_\mu 
-\frac{1}{8 \pi^2} \sum_{k \geq 0} f_{-2k}c_k'\Lambda^{-2k} \int_M \tr_N \hat F^{\mu\nu}  (-\Delta)^k \hat F_{\mu\nu}.
\end{align*}
where we have also restored the $\Lambda$-dependence by replacing $t \mapsto t \Lambda^{-2}$.
The first term cancels against Equation \eqref{eq:mass-term}. The explicit form of the coefficients follows from a combination of Lemma \ref{lem:F-term} and Lemma \ref{lem:kernel-term}:
$$
c_k' =  \frac{1}{2} \frac{k!}{(2k+1)!} -  \frac{(k+1)!}{(2k+3)!} = \frac{(k+1)!}{(2k+3)(2k+1)!}
$$
which is positive. It combines with $1/8 \pi^2$ to give the $c_k$'s displayed above. This completes the proof of Theorem \ref{thm:free-sa}. \endproof

This result could probably also be obtained by exploiting the leading terms expansion obtained in \cite{BGO90,Avr91}, after a careful counting of the number of contractions in $\nabla^{k/2} F \cdot \nabla^{k/2} F$ appearing in {\it loc.cit.}.

The formula for $S_0[A]$ can be checked with the Yang--Mills term appearing in \cite{CC96} (cf. Theorem \ref{thm:CC} above).
\begin{corl}
Modulo negative powers of $\Lambda$, we have
$$
S_0[A] \sim - \frac{f_0}{24 \pi^2} \int_M  \tr \hat F^{\mu\nu}  \hat F_{\mu\nu} + \cO(\Lambda^{-1}). 
$$
\end{corl}

We see that the first term in $S_0[A]$ is the usual (free part of the) Yang--Mills action. In fact, we could adjust the positive function $f$ so that $f_0 c_0 =1/4$. We end this section by introducing an expansion in $\Lambda$:
$$
\varphi_\Lambda(\Delta) =  \sum_{k \geq0} (-1)^k \Lambda^{-2k} f_{-2k} c_k  \Delta^k,
$$
so that we can write more concisely
$$
S^\Lambda_0[A] = - \int \tr \hat F_{\mu\nu} \varphi_\Lambda(\Delta) (\hat F^{\mu\nu}).
$$
This form motivates the interpretation of $S^\Lambda_0[A]$ (and of $S^\Lambda[A]$) as a higher-derivative gauge theory. As we will see below, this indeed regularizes the theory in such a way that $S^\Lambda[A]$ defines a superrenormalizable field theory.  

\begin{rem}
\label{rem:ILV}
Even though the above expansion $S_0^\Lambda$ is asymptotic for large $\Lambda$, it is interesting to consider the corresponding actual sum that defines $\varphi_\Lambda$. In particular, this allows to confront our results once again with \cite{ILV11}, by considering the large momentum limit of the full sum. Thus, consider
$$
\varphi_\Lambda (x) = \frac{1}{8 \pi^2} \int_{t \geq 0}  \sum_{k=0}^\infty \frac{(k+1)!}{ (2k+3)(2k+1)!} (-tx/ \Lambda^{2})^k  d \mu(t).
$$
One finds that
$$
\varphi_\Lambda(x) = \frac 1{8 \pi^2} \int \left( \left( \frac{\Lambda}{\sqrt{tx}} +2  \frac{\Lambda^3}{(tx)^{3/2}} \right) F\left( \frac{\sqrt{tx}}{2\Lambda} \right) - \frac{\Lambda^2}{tx}\right) d \mu(t).
$$
Here, the Dawson function $F$ is defined in terms of the error function by 
$$
F(z) = \frac{\sqrt{\pi}}2 e^{-z^2} \textup{erfi}(z).
$$
The asymptotic behaviour of $\varphi_\Lambda(p^2)$ as $p^2 \to \infty$ can then be determined to be 
$$
\varphi_\Lambda(p^2) \underset{p^2 \to \infty}{\sim} \frac{1}{2 \pi^2} f_{4}  \Lambda^4 p^{-4} + \cdots
$$
using that $F(z) \sim 1/2z + 1/4z^3 + \cdots$ as $z \to \infty$. It is striking that already at this heuristic level it would lead to the same conclusion on the UV-behaviour of the spectral action as in \cite{ILV11}. 
A fully rigorous analysis of $\varphi_\Lambda(p^2)$ can be done by evaluating the integrals over $s$ appearing in the proofs of the above Lemmas \ref{lem:F-term} and \ref{lem:kernel-term} before asymptotically expanding the heat kernel using Eq. \eqref{heat-kernel-expansion}. This indeed confirms the results of \cite{ILV11} (in particular their Equation (25)).

We avoid such behaviour of the gauge propagator by working with the asymptotic expansion $S^\Lambda[A]$ in large $\Lambda$. Together with a suitable choice of the function $f$, it is precisely this expansion which allows us to obtain polynomial growth for $\varphi_\Lambda(p^2)$ for large $p$. This allows us to show that the gauge field theory defined by $S^\Lambda[A]$ is (super)renormalizable, as we will now proceed to discuss.
\end{rem}

\section{Gauge fixing in the YM-system}

We add a gauge-fixing term of the following higher-derivative form:
\begin{equation}
\label{sa-gf}
S^\Lambda_\gf[A]= - \frac{1}{2 \xi}  \int \tr_N \partial_\mu A^\mu \varphi_\Lambda(\Delta) \left( \partial_\nu A^\nu \right) .
\end{equation}
In order to derive the gauge propagator, we need to invert the quadratic form given by $S^\Lambda_0[A] + S^\Lambda_\gf[A]$. This is only possible if $\varphi_\Lambda(p^2)$ is nonvanishing, in which case it is given by
$$
D_{\mu\nu}^{ab}(p; \Lambda) = \left[ g_{\mu\nu} - (1-\xi) \frac{p_\mu p_\nu}{p^2}\right] \frac{\delta^{ab}}{p^2 \varphi_\Lambda(p^2)} .
$$
The non-vanishing of $\varphi_\Lambda(p^2)$ can be guaranteed by assuming that $f^{(2k)}(0) \geq 0$.
Indeed, since by Lemma \ref{lma:coef} we have that $f_{-2k} = (-1)^k f^{(2k)}(0)/(2k-1)!!$, in that case the summands constituting $\varphi_\Lambda$ are positive so that indeed $\varphi_\Lambda(p^2) \neq 0$. In the following, we will make the above assumptions on the even higher derivatives of $f$ at zero. The behaviour of $D_{\mu\nu}^{ab}$ for large $p$ will be discussed in more detail in the next section.

As usual, the above gauge fixing requires a Jacobian, conveniently described by a Faddeev--Popov ghost Lagrangian:
\begin{equation}
\label{sa-gh}
S^\Lambda_\gh[A,\bar C,C] = - \int \tr_N \partial_\mu \bar C \varphi_\Lambda(\Delta) \left( \partial^\mu C + [A^\mu,C] \right).
\end{equation}
Here $C,\bar C$ are the Faddeev--Popov ghost fields and their propagator is
$$
\tilde D^{ab}(p; \Lambda) = \frac{\delta^{ab}}{p^2 \varphi_\Lambda(p^2)}.
$$

\begin{prop}
The sum $S^\Lambda[A] + S^\Lambda_\gf[A] + S^\Lambda_\gh[A,\bar C, C]$ is invariant under the BRST-transformations:
\begin{gather}
\label{brst}
sA_\mu = \partial_\mu C + [A_\mu,C];\qquad s C = -\half [C,C]; \qquad s \bar C =  \xi^{-1} \partial_\mu A^\mu.
\end{gather}
\end{prop}
\begin{proof}
First, $s(S)=0$ because of gauge invariance of $S[A]$. We compute 
\begin{align*}
s(S^\Lambda_\gf) &= -\frac{1}{\xi} \int \tr_N ( \partial_\mu A^\mu) \varphi_\Lambda(\Delta) \left( \partial_\nu \partial^\nu C + \partial_\nu( [A^\nu,C] \right) .
\intertext{On the other hand,}
s(S^\Lambda_\gh) &= - \frac{1}{\xi} \int \tr_N (\partial_\mu \partial^\nu A_\nu )\varphi_\Lambda(\Delta) \left( \partial^\mu C+ [A^\mu,C] \right)
\end{align*}
which modulo vanishing boundary terms is minus the previous expression.
\end{proof}
Note that $s^2 \neq 0$, which can be cured by standard homological methods: introduce an auxiliary (aka Nakanishi-Lautrup) field $h$ so that $\bar C$ and $h$ form a contractible pair in BRST-cohomology. In other words, we replace the above transformation in \eqref{brst} on $\bar C$ by $s \bar C = -h$ and $s h = 0$. If we replace $S^\Lambda_\gf + S^\Lambda_\gh$ by $s \Psi^\Lambda$ with $\Psi^\Lambda$ an arbitrary {\it gauge fixing fermion}, it follows from gauge invariance of $S^\Lambda$ and nilpotency of $s$ that $S^\Lambda + s \Psi^\Lambda$ is BRST-invariant. The above special form of $S^\Lambda_\gf+ S^\Lambda_\gh$ can be recovered by choosing
$$
\Psi^\Lambda =-  \int \tr_N \varphi_\Lambda(\Delta) (\bar C) \left( \half \xi h + \partial_\mu A^\mu \right).
$$

\begin{rem}
\label{rem:gf}
One might wonder what gauge fixing condition is implemented by $S^\Lambda_\gf$ as in \eqref{sa-gf}, given the presence of the term $\varphi_\Lambda(\Delta)$. If $f^{(2k)}(0) \geq 0$, then the function $x \mapsto \varphi_\Lambda(x^2)$ is positive, turning the bilinear form
$$
(\omega_1,\omega_2) := - \int \tr_N \omega_1 \wedge \ast (\varphi_\Lambda(\Delta) \omega_2)
$$
into an inner product. On the Lagrangian level, we can equally well implement the Lorenz gauge fixing condition $\partial \cdot A = 0$ using this inner product instead of the usual $L^2$-inner product. This gives rise to $S^\Lambda_\gf[A] =  ( \partial \cdot A, \partial \cdot A)/2\xi$. Similarly, $S^\Lambda_\gf$ is given by the inner product $(\bar C, \partial_\mu C + [A_\mu,C])$.
\end{rem}

\section{Renormalization of the asymptotically expanded YM-spectral action}
\label{sect:ren-sa}
As said, we consider the asymptotically expanded spectral action for the Yang--Mills system as a higher-derivative field theory. This means that we will use the higher derivatives of $F_{\mu\nu}$ that appear in the asymptotic expansion as natural regulators of the theory, similar to \cite{Sla71,Sla72b} (see also \cite[Sect. 4.4]{FS80}). However, note that the regularizing terms are already present in the asymptotic expansion $S^\Lambda[A]$ of the spectral action and need not be introduced as such. 
Let us consider the expansion of Theorem \ref{thm:free-sa} up to order $n$ (which we assume to be at least $8$), {\it i.e.} we set $f_{4-m} = 0$ for all $m > n$ while $f_0, \ldots f_{4-n} \neq 0$. Also, assume a gauge fixing of the form \eqref{sa-gf} and \eqref{sa-gh}.

Then, we easily derive from the structure of $\varphi_\Lambda(p^2)$ that the propagators of both the gauge field and the ghost field behave as $|p|^{-n+2}$ as $|p| \to \infty$. Indeed, in this case:
$$
\varphi_\Lambda(p^2) = \sum_{k=0}^{n/2-2} (-1)^k \Lambda^{-2k} f_{-2k} c_k p^{2k}.
$$
Moreover, the weights of the interaction in terms of powers of momenta is given by:
$$
\begin{array}{|c|c|c|}
\hline
\text{vertex} & \text{valence} & \max \# \text{ der}\\
\hline
\hline
\parbox{2cm}{\vspace{1mm}\includegraphics[scale=.2]{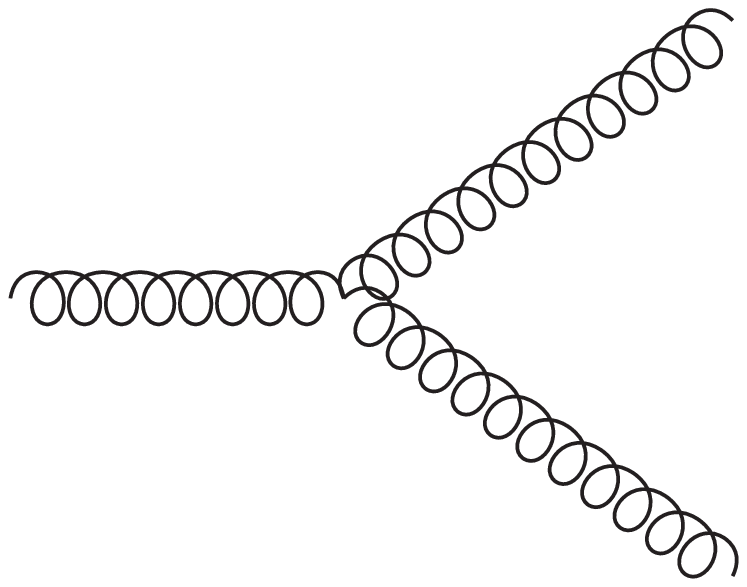}} & 3 & n-3\\[2mm]
\parbox{2cm}{\includegraphics[scale=.2]{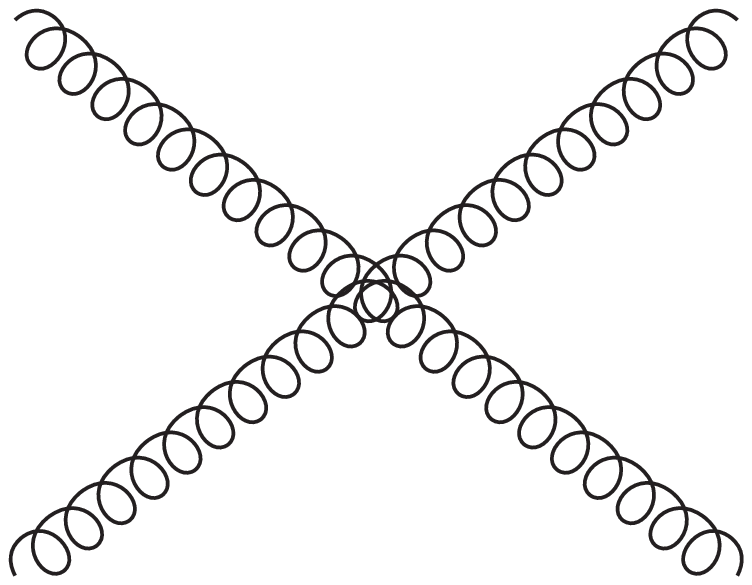}} & 4 & n-4\\[2mm]
\vdots& \vdots & \vdots\\
\parbox{2cm}{\includegraphics[scale=.2]{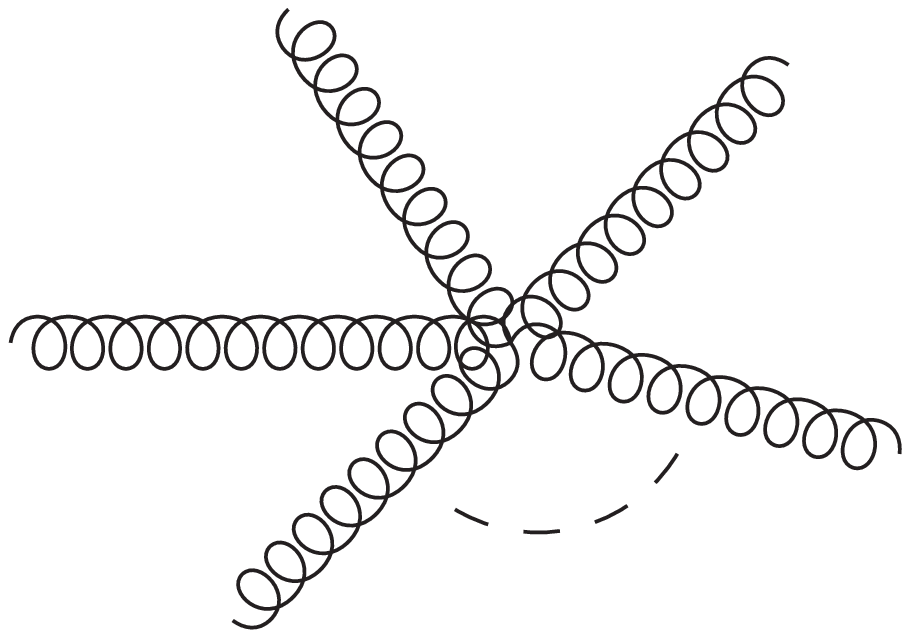}} & n & 0 \\[2mm]
\parbox{2cm}{\includegraphics[scale=.2]{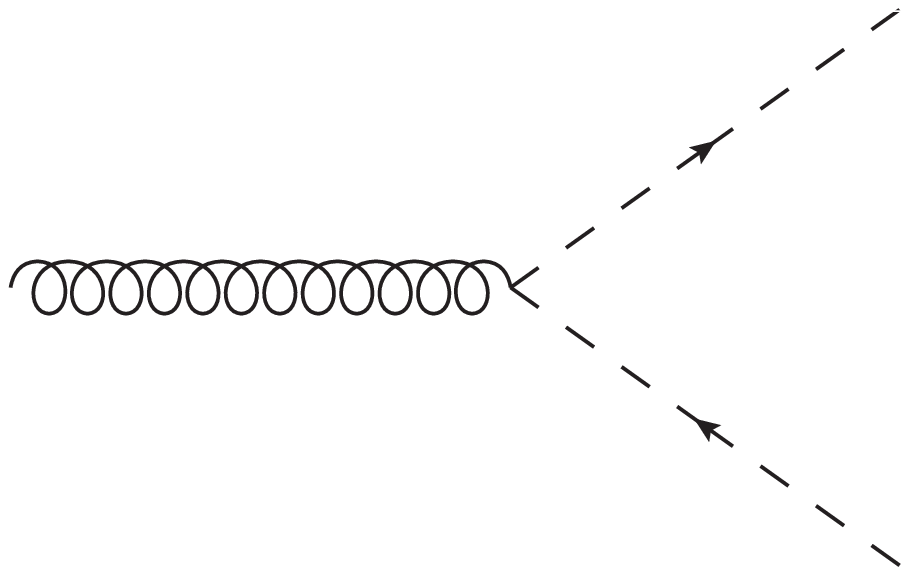}} & 3 & n-3\\[2mm]
\hline
\end{array}
$$
We will use $v_k$ to indicate the number of gauge interaction vertices of valence $k$, and with $\tilde v$ the number of ghost-gauge vertices.

Let us now find an expression for the {\it superficial degree of divergence} $\omega$ of a Feynman graph consisting of $I$ internal gauge edges, $\tilde I$ internal ghost edges, $v_k$ valence $k$ gauge vertices and $\tilde v$ ghost-gauge vertices. In 4 dimensions, we find at loop order $L$:
$$
\omega \leq 4L - I(n-2) - \tilde I (n-2) + \sum_{i=3}^n v_i (n-i) + \tilde v (n-3).
$$
\begin{lma}
Let $E$ and $\tilde E$ denote the number of external gauge and ghost edges, respectively. The superficial degree of divergence of the graph satisfies:
$$
\omega \leq (4-n)(L-1) + 4 - (E+\tilde E).
$$
\end{lma}
\begin{proof}
We use the relations
$$
2 I + E = \sum_i i v_i + \tilde v; \qquad
2 \tilde I + \tilde E = 2\tilde v 
$$
where $E$ and $\tilde E$ are the number of external gauge and ghost legs, respectively. Indeed, these formulas count the number of half (gauge/ghost) edges in a graph in two ways: from the number of edges and from the valences of the vertices. We use them to substitute for $2I$ and $2\tilde I$ in the above expression for $\omega$ so as to obtain
$$
\omega = 4L - In - \tilde I n + n \left(\sum_i v_i + \tilde v \right) - (E + \tilde E)
$$
from which the result follows at once from Euler's formula $L= I + \tilde I - \sum_i v_i - \tilde v +1$.
\end{proof}
As a consequence, $\omega < 0$ if $L \geq 2$ (provided $n \geq 8$): all Feynman graphs are finite at loop order greater than 1. If $L=1$, then there are finitely many graphs which are divergent, namely those for which $E+ \tilde E \leq 4$. We conclude that the asymptotically expanded spectral action for the Yang--Mills system is superrenormalizable. 

Of course, the spectral action being a gauge theory, there is more to renormalizability than just power counting: we have to establish gauge invariance of the counterterms. 
We already know that the counterterms needed to render the perturbative quantization of $S^\Lambda[A]$ finite are of order $4$ or less in the fields and arise only from one-loop graphs. The key property of the effective action at one loop is that it is supposed to be BRST-invariant:
$$
s(\Gamma_{1}) = 0. 
$$
In particular, assuming a regularization compatible with gauge invariance, the divergent part $\Gamma_{1,\infty}$ is BRST-invariant. Results from \cite{Dix91, DTV85, DTV85b, BDK90,DHTV91} on BRST-cohomology for Yang--Mills type theories ascertain that the only BRST-closed functional of order 4 or less in the fields is represented by
$$
\delta Z \int F_{\mu\nu}F^{\mu\nu} 
$$
for some constant $\delta Z$. 
We will confirm this through an explicit calculation using zeta function regularization in background field gauge in the remaining part of this paper. 

This particular form for the counterterm $\Gamma_{1,\infty}$ can be added to $S^\Lambda$ and absorbed by a redefinition of the function $f$. Indeed, it maps 
$$
f_0 \mapsto f_0 + 24 \pi^2\delta Z
$$
leaving all other coefficients $f_{-2k}$ ($k \neq 0$) invariant. Intriguingly, renormalization of $S^\Lambda[A]$ can thus be accomplished merely by shifting the function $f$ by a constant amount $24 \pi^2\delta Z$.

\section{One-loop effective action}
\label{sect:1loop-sa}

\subsection{Background field gauge}
Adopting the background field formalism, one expands $S^\Lambda[A]$ around a (skew-hermitian) background gauge field $B$. For the one-loop effective action it is sufficient to consider only the quadratic part in $A$ (skew-hermitian), for which we compute 
$$
S^\Lambda_0[A;B] = - \sum_{n,m \geq 0} \Lambda^{2-n-m} f_{2-n-m} \int_M \tr_N A_\mu c^{(n)\mu\nu}_{\mu_1\cdots \mu_m} \nabla^{\mu_1}_B \cdots \nabla^{\mu_m}_B A_\nu .
$$
where $c^{(n)}$ are gauge invariant and Lorentz covariant polynomials in $B_\mu$ of order $n$. 
For example, in terms of the $c_k$ from Theorem \ref{thm:free-sa}:
\begin{align}
&c^{(0)\mu\nu}_{\mu_1\cdots \mu_m} = (-1)^{m/2} 2 c_{(m-2)/2}\left(
g^\mu_{\phantom\mu \mu_1} g_{\mu_2 \mu_3} \cdots g_{\mu_{m-2} \mu_{m-1}} g_{\mu_m}^{\phantom{\mu_m}\nu}-
g^{\mu \nu} g_{\mu_1 \mu_2} \cdots g_{\mu_{m-1} \mu_m}  \right)
\label{eq:top-coef}
\intertext{and (less explicitly for some constants $\alpha,\beta,\gamma$ and $\delta$)}
&c^{(2)\mu\nu}_{\mu_1\cdots \mu_m} = \alpha F^{\mu \nu} g_{\mu_1 \mu_2} \cdots g_{\mu_{m-1} \mu_m} 
+ \beta \sum_i F^{\mu \mu_i} g_{\mu_1\mu_2} \cdots g_{\nu \mu_{i\pm1}} \cdots g_{\mu_{m-1} \mu_m} \nn \\ \nn
& \quad+ \gamma \sum_i F^{\nu \mu_i} g_{\mu_1\mu_2} \cdots g_{\mu \mu_{i\pm1}} \cdots g_{\mu_{m-1} \mu_m}+ \delta \sum_{i\neq j} F^{\mu_i \mu_j} g_{\mu_1\mu_2} \cdots g_{\mu \mu_{i\pm1}}  \cdots g_{\nu \mu_{j\pm 1}} \cdots g_{\mu_{m-1} \mu_m}.
\end{align}
In the following, we will lighten notation and absorb the power $\Lambda^{2-n-m}$ in the coefficient $f_{2-n-m}$. 

The background field gauge is defined by setting $\nabla^\mu_B A_\mu = 0$ which we accomplish as above (cf. Remark \ref{rem:gf}) through the inner product with weight $\varphi_\Lambda(\Delta_B)$:
\begin{align*}
S^\Lambda_\gf[A;B]&:=-\frac{1}{2 \xi} \int_M (\nabla_B^\mu A_\mu) \varphi_\Lambda(\Delta_B) (\nabla_B^\nu A_\nu) \\
&=
-2\xi^{-1} \sum_{m \geq 2} (-1)^{m/2} f_{2-m} c_{(m-2)/2} \int_M (\nabla_B^\mu A_\mu) \Delta_B^{m-2} (\nabla_B^\nu A_\nu).
\end{align*}
For instance, for $\xi=1$ (Feynman gauge), it precisely cancels the first `longitudinal' term in \eqref{eq:top-coef}; compare with Equation \eqref{sa-gf}, where $B=0$. We define a differential operator $P_B$ by setting
$$
S^\Lambda_0[A;B] + S^\Lambda_\gf[A;B] =: \int_M \tr A_\mu P_B^{\mu\nu} (A_\nu). 
$$

\subsection{One-loop effective action}
The background field formalism is very convenient to compute the effective action at one-loop. Namely, it is given by the determinant
$$
W [B] = \frac{1}{2} \ln \det (P_B P^{-1})
$$
where $P \equiv P_{B=0}$. As before, we terminate the expansion defining $P_B$ by assuming $f_{2-m}$ is zero for all $m>2k$ ($k$ is related to the integer $n$ in Section \ref{sect:ren-sa} via $2k = n-2$). This makes $P_B$ a differential operator of order $2k$ and for $\xi=1$ (Feynman gauge) can be written in the form 
$$
P_B^{\mu\nu} =  f_{2-2k} g^{\mu\nu} \Delta^k + p_{2,\alpha\beta}^{\mu\nu} \nabla^\alpha \nabla^\beta \Delta^{k-2} +  (-1)^k \left \{  p_{3,\alpha_1\cdots \alpha_{2k-3}}^{\mu\nu} \nabla^{\alpha_1} \cdots \nabla^{\alpha_{2k-3}} + 
 \cdots \right\}
$$
with 
$$
p_{2,\alpha\beta}^{\mu\nu} =  f_{2-2k} \left(\alpha F^{\mu\nu} g_{\alpha\beta} + \beta' F^{\mu}_{\phantom\mu(\alpha} g^{\nu}_{\phantom\nu \beta)} + \gamma'  F^{\nu}_{\phantom\mu(\beta} g^{\mu}_{\phantom\nu \alpha)} 
  \right)
  +  2 f_{2-2(k-1)} c_{(2k-4)/2} g^{\mu\nu}g_{\alpha\beta}
$$
and $p_{4,\alpha_1 \cdots \alpha_{k-4}}$ are gauge invariants of $B$ (of order less then or equal to 2 and 4, respectively). The invariant $p_4$ also depends linearly on $f_{2-2k}$, $ f_{2-2(k-1)}$ and $ f_{2-2(k-2)}$.

The above determinant is ill-defined but we make sense of it in the following way \cite{Haw77}. First, we (still heuristically) apply the ``$\ln\det= \tr \ln$'' formula to obtain
$$
\ln \det P_B P^{-1} = - \int^\infty_0 \frac{dt}{t} \tr (e^{-t P_B} - e^{-t P}).
$$
Indeed, this follows simply from the fact that
$$
\ln \lambda = - \int_0^\infty \frac{dt}{t} (e^{-t\lambda} -e^{-t}).
$$
Zeta function regularization is the procedure to replace the determinant by
$$
\frac{1}{2} \ln \det P_B P^{-1} := -\frac{1}{2} \tilde{\mu}^{2 k z} \int^\infty_0 \frac{dt}{t^{1-z}} \tr (e^{-t P_B}- e^{-t P} )
$$
with $z\in \C$ (see \cite{Vas03} for an excellent review). Here $\tilde\mu$ is the so-called {\it mass scale}, introduced to keep the effective action dimensionless. Now, with the Mellin transform
\begin{equation}
\label{eq:mellin}
\Gamma(z) \lambda^{-z} = \int_0^\infty dt ~t^{z-1} e^{-t\lambda}; \qquad (\Re z>0),
\end{equation}
we find that the regularized one-loop effective action reads
\begin{equation}
\label{eq:eff-action-reg}
W_z [B] = -\frac{1}{2} \tilde{\mu}^{2kz} \Gamma(z) (\tr P_B^{-z} - \tr P^{-z})  \equiv  -\frac{1}{2}  \tilde{\mu}^{2kz}\Gamma(z)( \zeta(P_B,z) - \zeta(P,z)).
\end{equation}
As a function of $z$ this expression is holomorphic for $\Re z >>0$. It has a simple pole at $z=0$:
$$
W_z [B] = -\frac{1}{2}(  \zeta(P_B,0) - \zeta(P,0))\left(\frac{1}{z} + 2k \ln \mu\right) -\frac{1}{2}( \zeta'(P_B,0)- \zeta'(P,0)) + \cO(z)
$$
which follows by expanding the gamma function (further absorbing the Euler constant $\gamma_E$ in $\tilde \mu$ to define $\mu$ by $\mu^{2k} = e^{-\gamma_E} \tilde \mu^{2k}$). The first term is the counterterm that should be subtracted from the spectral action. Let us determine its form in terms of the curvature of the background field $B$.
We apply a result by Gilkey \cite{Gil80} (cf. Theorem \ref{thm:gilkey} below) to calculate
$$
\zeta(P_B,0) = a_4 (P_B)
$$
since the differential operator $P_B$ is of the form \eqref{eq:higher-laplacian} and we are in 4 dimensions. 
Since we have assumed $M$ is flat, the expression for $a_4(P_B)$ reduces to
\begin{align*}
a_4(P_B) &=\frac{1}{4 \pi^2}
\int_M \tr_N \bigg(
\frac{1}{12} F_{\mu\nu}F^{\mu\nu}
+ \frac{1}{48 k} \left( {p_{2,\alpha}}^\alpha  {p_{2,\beta}}^\beta 
+2 {p_{2,\alpha}}^\beta  {p_{2,\beta}}^\alpha \right)
-\frac{1}{k S(\delta^{k-2})} S(p_4)
\bigg) .
\end{align*}
Since the $p_2$ and $p_4$ are gauge invariants depending on $B$, with maximal order equal to 2 and 4, respectively, we find that $a_4$ is a gauge invariant functional of $B$, of order less then or equal to $4$. 

Taking into account that $a_4(P_B)$ is scale invariant (the coefficient of $t^0$ in the heat expansion) we arrive at the following result.

\begin{prop}
\label{prop:eff-action-gauge}
There exist a constant $c$ 
such that the residue of the regularized effective action is
$$
-\frac{1}{2}( \zeta(P_B,0) -\zeta(P,0)) = c 
\int_M \tr F_{\mu\nu}F^{\mu\nu}. 
$$
\end{prop}
\subsection{Gauge fixing and the Faddeev--Popov action}
The Jacobian term for the background field gauge can be implemented by adding a Faddeev--Popov action. In background gauge it becomes
$$
S^\Lambda_\gh[\bar C,C; B] = \int \tr\bar C \varphi_\Lambda(\Delta_B) \nabla_\mu^B  (\nabla^\mu_B+ \ad A^\mu) C.
$$
Introducing also the auxiliary (hermitian) field $h$ as before, we replace $S^\Lambda_\gf[A;B]$ by 
$$
S_\gf[A,h;B] =  \int \tr h ~\varphi_\Lambda(\Delta_B) (\half \xi h + \nabla^B_\mu A^\mu) .
$$
Indeed, 
$$
S^\Lambda_\gf[A,h;B] = \frac{1}{2\xi} \int \tr \left ( \xi h + \nabla^B_\mu A^\mu) \varphi_\Lambda(\Delta_B) ( \xi h + \nabla^B_\mu A^\mu) \right) -\frac{1}{2 \xi} \int_M (\nabla_B^\mu A_\mu) \varphi_\Lambda(\Delta_B) (\nabla_B^\nu A_\nu)
$$
so that the first term describes a free field, decoupled form the $A$ and $B$ field, and the second term is our previous expression for $S^\Lambda_\gf[A;B]$. 
\begin{prop}
The sum $S^\Lambda[A+B] + S^\Lambda_\gf[A,h;B] + S^\Lambda_\gh[\bar C,C,A;B]$ is invariant under the following BRST-transformations:
\begin{gather}
\label{eq:BRST-bfg}
s(A_\mu) = \nabla_\mu^B C + [A_\mu , C], \qquad s(C) = -\frac{1}{2} [C,C] ,\qquad
s(\bar C) = -h, \qquad s h = 0. 
\end{gather}
\end{prop}
\begin{proof}
The effect of $s$ on $A_\mu+B_\mu$ is just a gauge transformation, whence leaving $S^\Lambda[A+B]$ invariant. For the gauge-fixing term, we compute:
$$
s\left(  \int \tr h ~\varphi_\Lambda(\Delta_B) (\half \xi h + \nabla^B_\mu A^\mu) \right) =  \int \tr h ~\varphi_\Lambda(\Delta_B) (\nabla^B_\mu \nabla^\mu_B C + [A^\mu,C])
$$
which cancels against the first term in:
\begin{align*}
s\left(\int \tr \bar C \varphi_\Lambda(\Delta_B) \nabla_\mu^B  (\nabla^\mu_B+ \ad A^\mu) C\right) &= - \int \tr  h \varphi_\Lambda(\Delta_B) \nabla_\mu^B  (\nabla^\mu_B+ \ad A^\mu) C \\
&\quad - \int \tr \bar C \varphi_\Lambda(\Delta_B) \nabla_\mu^B  s\left((\nabla^\mu_B+ \ad A^\mu) C\right).
\end{align*}
The vanishing of the term $s\left((\nabla^\mu_B+ \ad A^\mu) C\right)$ is an easy consequence of Jacobi's identity and the Leibniz rule.
\end{proof}

The contribution of the Faddeev--Popov ghost action to the one-loop effective action can  be obtained along the same lines as above. First, write
$$
S^\Lambda_\gh[\bar C,C;B] = - \int_M \tr \bar C \tilde P_B( C)
$$
where $\tilde P_B$ is a differential operator of order $2k$ and of the form
\begin{multline*}
\tilde P_B =  f_{2-2k} \Delta^k + ( f_{2-2k} +  f_{2-2(k-1)}) \Delta^{k-1} \\+  p_{3,\alpha_1\cdots \alpha_{k-3}} \nabla^{\alpha_1} \cdots \nabla^{\alpha_{k-3}} + p_{4,\alpha_1\cdots \alpha_{k-4}} \nabla^{\alpha_1} \cdots \nabla^{\alpha_{k-4}} + \cdots
\end{multline*}
where $p_3$ and $p_4$ are gauge invariants functions of $B$. Since $\bar C$ and $C$ are fermionic, their contribution to the effective action is $-\ln \det \tilde P_B \tilde P^{-1}$, regularized using zeta functions as above. Gilkey' s Theorem applies to arrive at
\begin{prop}
\label{prop:eff-action-ghost}
There exist a constant $\tilde c$ 
such that the residue of the ghost contribution to the regularized effective action is
$$
\zeta(\tilde P_B,0) = \tilde c 
\int_M \tr F_{\mu\nu}F^{\mu\nu} 
$$
\end{prop} 

\subsection{Renormalized YM-spectral action}

Combining the above two Propositions \ref{prop:eff-action-gauge} and \ref{prop:eff-action-ghost} we can derive an expression for the renormalized asymptotically expanded spectral action for the Yang--Mills system. Adding counterterms to $S^\Lambda[A]$ gives
$$
S^\Lambda_\ren[A] = S^\Lambda[A] + \frac{1}{2} \left(\frac{1}{z}+2k \ln \mu \right) ( \zeta(P_A,0) - \zeta(P,0))  -  \left(\frac{1}{z}+2k \ln \mu \right) (\zeta(\tilde P_A,0) - \zeta(\tilde P,0)).
$$
Recall from Theorem \ref{thm:CC} the lowest order terms in the asymptotic expansion $S^\Lambda[A]$. 
\begin{thm}
The action $S^\Lambda[A]$ for the Yang--Mills system on a flat background manifold is renormalizable. The renormalized (asymptotically expanded) spectral action $S^\Lambda_\ren[A]$ is obtained from $S^\Lambda[A]$ by the following redefinition of the coefficient $f_0$:
\begin{gather*}
f_0 \mapsto f_0 + 24 \pi^2 \left(c +\tilde c\right)
\left(\frac{1}{z} + 2k \ln \mu\right)  
\end{gather*}
where the constants $c,\tilde c$ are 
as in Proposition \ref{prop:eff-action-gauge} and \ref{prop:eff-action-ghost}.
\end{thm}
With this result, we could determine $\beta$ functions for the dependence of the renormalized quantity $f_0$ 
on the mass scale $\mu$. Defining a bare quantity $f_0^B$ 
by 
\begin{gather*}
f_0^B:=  f_0 +24 \pi^2 \left(c +\tilde c\right)
\left(\frac{1}{z} + 2k \ln \mu\right),  
\end{gather*}
its supposed independence of the mass scale $\mu$ implies that
$$
\mu \frac{\partial f_0}{\partial \mu} = - 48 k \pi^2 \left(c +\tilde c\right).
$$
This defines the renormalization group flow on the spectral function $f$, which in this case is linear on the coefficient $f_0$. 
Also, since the only divergences appear at one loop, these equations are expected to hold at the non-perturbative level.

\section{Explicit computation of the counterterm for $k=2$}
\label{sect:1l-sa-a6}
The goal is to compute the counterterm for the asymptotically expanded spectral action with $f_{2-m}=0$ for all $m>4$. In other words, 
$$
S^\Lambda[A] = f_0 a_4(D_A^2) + \Lambda^{-2} f_{-2} a_6 (D_A^2).
$$
Even though the corresponding higher-derivative theory is not superrenormalizable (that would require at least $f_{-4} \neq 0$), it serves as an illustrative example of the above approach. 

For convenience, we compute in a covariantly flat background field, {\it i.e.}
$$
\nabla_\mu (F_{\nu\rho}) = 0
$$
where we have dropped the index $B$ from the covariant derivative and curvature. Consequently, the $F_{\mu\nu}$ commutes among themselves. One then easily establishes the following results:
\begin{lma}
\label{lma:rel1}
For a covariantly flat connection $\nabla$ we have for any field $\phi$ in the adjoint representation:
\begin{align}
\tag{a}\label{rel1-a}
\Delta \nabla_\mu \phi &= \nabla_\mu \Delta \phi + 2 [F_{\mu\kappa},\nabla^\kappa \phi],\\
\tag{b}\label{rel1-b}
\nabla_\mu \Delta \nabla_\nu \phi &= \nabla_\nu \Delta \nabla_\mu \phi +[F_{\mu\nu},\Delta \phi] + 2 [F_{\nu\kappa},\nabla_\mu \nabla^\kappa \phi] - 2 [F_{\mu\kappa},\nabla_\nu \nabla^\kappa \phi],\\
\tag{c}\label{rel1-c}
-\nabla_\mu \Delta \nabla^\mu \phi & = \Delta^2 \phi - [F^{\mu\kappa},[F_{\mu\kappa}, \phi]].
\end{align}
with $\Delta = -\nabla^\mu \nabla_\mu$ the Laplacian corresponding to $\nabla$.
\end{lma}

We aim at expressions for the quadratic part of our action functionals of the following form
\begin{equation}
\label{P_B}
-\tfrac{1}{2} \int \tr A_\mu P_{B}^{\mu\nu} A_\nu; \qquad P_B ^{\mu\nu} = \Delta^2 g^{\mu\nu}+ p_{2,\alpha\beta}^{\mu\nu} \nabla^\alpha \nabla^\beta + p_4^{\mu\nu}
\end{equation}
with $p_2$ and $p_4$ endomorphisms depending on the background field $B$. We first examine some other related action functionals that will be useful in studying $S^\Lambda[A]$.

\begin{prop}
\label{prop:S1}
Consider the action functional
$$
S_1[A] = -\half\int \tr {F^{\mu\nu}}_{;\mu}(A) {F_{\rho \nu}}^{;\rho}(A)-\half \int \tr (\nabla_\nu A^\nu) \Delta \nabla_\mu A^\mu .
$$
Then the quadratic part in $A$ in a covariantly flat background $B$ is of the form \eqref{P_B} with
\begin{align*}
p_{2,\alpha\beta}^{\mu\nu} = 4 \ad F^{\mu\nu} g_{\alpha\beta},\qquad
p_4^{\mu\nu} = -4 \ad F^{\mu\kappa} \ad {F^\nu}_\kappa,
\end{align*}
in terms of the curvature of $B$.
\end{prop}
\proof
The quadratic part in $A$ for a covariantly flat background field $B$ is 
\begin{align*}
S_1[A;B] &\equiv S_1[B+A]|_{\textrm{quad}} \\
& = -\half \int \tr (\nabla_\nu A^\nu) \Delta \nabla_\mu A^\mu - \half \int \tr \left( [A_\mu,F^{\mu\nu}] - \Delta A^\nu - \nabla_\mu \nabla^\nu A^\mu \right) \\
& \qquad \qquad \qquad \times\left( [A^\rho,F_{\rho\nu}] - \Delta A_\nu - \nabla^\rho \nabla_\nu A_\rho \right)\\
&=-\half \int \tr \bigg\{ A^\nu \Delta^2 A_\nu - A_\mu [F^{\mu\nu},[F_{\rho\nu},A^\rho] ]-2 A_\mu [F^{\mu\nu}, \Delta A_\nu]\\ 
&\qquad \qquad \qquad -2 A_\mu [F^{\mu\nu},\nabla^\rho \nabla_\nu A_\rho] +2A^\nu \Delta \nabla^\rho \nabla_\nu A_\rho + A^\mu \nabla^\nu \nabla_\mu \nabla^\rho \nabla_\nu A_\rho 
\bigg\}.
\end{align*}
We use the commutation relation \eqref{rel1-a} of Lemma \ref{lma:rel1} and the definition of the curvature to obtain for the last two terms:
\begin{align*}
2A^\nu \Delta \nabla^\rho \nabla_\nu A_\rho&=2 A^\nu [F_{\rho\nu},\Delta A^\rho] + 2 A^\nu \nabla_\nu \Delta \nabla^\rho A_\rho + 4 A^\nu [F_{\nu \kappa}, \nabla^\kappa \nabla^\rho A_\rho],
\\
A^\mu \nabla^\nu \nabla_\mu \nabla^\rho \nabla_\nu A_\rho &=A^\mu [F_{\rho\nu}, \nabla^\nu \nabla_\mu A^\rho] + A^\mu [{F^\nu}_\mu, \nabla_\nu \nabla^\rho A_\rho] + A^\mu \nabla_\mu \nabla^\nu \nabla_\nu \nabla^\rho A_\rho.
\end{align*}
Substituting this in $S_1[A;B]$ and relabeling dummy indices if necessary, we find
\begin{align*}
S_1[A;B] &= -\half \int \tr \bigg\{ A^\nu \Delta^2 A_\nu -A_\mu[F^{\mu\nu},[F_{\rho\nu},A^\rho]] -4 A_\mu [F^{\mu\nu},\Delta A_\nu] + 3 A^\mu [F_{\mu\kappa},[F^{\kappa\nu},A_\nu]] \bigg\}
\end{align*}
from which we read off the displayed form of $p_2$ and $p_4$.
\endproof

\begin{prop}
\label{prop:S2}
Consider the action functional
$$
S_2[A] = \tfrac{1}{4} \int \tr F_{\mu\nu}(A) {F^{\mu \nu;\kappa}}_{\kappa}(A) - \half \int \tr (\nabla_\nu A^\nu) \Delta \nabla_\mu A^\mu .
$$
Then the quadratic part in $A$ in a covariantly flat background $B$ is of the form \eqref{P_B} with
\begin{align*}
p_{2,\alpha\beta}^{\mu\nu} &= \ad F^{\mu\nu} g_{\alpha\beta} + \tfrac{3}{2}(\ad {F^\mu}_\beta g_\alpha^\nu + \ad {F^\mu}_\alpha g_\beta^\nu) -\tfrac{3}{2}( \ad {F^\nu}_\beta g^\mu_\alpha + \ad {F^\nu}_\alpha g^\mu_\beta),\\
p_4^{\mu\nu} &= -\tfrac{3}{2} \ad F_{\kappa\lambda} \ad F^{\kappa\lambda} - \ad F^{\mu \kappa} \ad {F^\nu}_\kappa,
\end{align*}
in terms of the curvature of $B$.
\end{prop}
\proof
The quadratic part in $A$ for a covariantly flat background field $B$ is 
\begin{align*}
S_2[A;B] &\equiv S_2[B+A]|_{\textrm{quad}} \\
&=-\half \int \tr (\nabla_\nu A^\nu) \Delta \nabla_\mu A^\mu  - \tfrac{1}{4} \int \tr \left([A_\kappa,F_{\mu\nu}] + \nabla_\kappa (\nabla_\mu A_\nu- \nabla_\nu A_\mu) \right)^2
\\
&=-\half \int \tr (\nabla_\nu A^\nu) \Delta \nabla_\mu A^\mu  - \half \int \tr \bigg\{ -\half A_\kappa[F_{\mu\nu},[F^{\mu\nu},A^\kappa]]\\
&\qquad\qquad\qquad \qquad\qquad
-2A^\nu [F_{\mu\nu}, \nabla^\mu \nabla^\kappa A_\kappa] - A^\nu \nabla^\mu \Delta \nabla_\mu A_\nu + A^\nu \nabla^\mu \Delta \nabla_\nu A_\mu \bigg\}\\
&= -\half \int \tr \bigg\{ -\tfrac{3}{2} A_\kappa[F_{\mu\nu},[F^{\mu\nu},A^\kappa]]
+A^\nu \Delta^2 A_\nu + A^\nu [F_{\mu\nu},\Delta A^\mu] \\
&\qquad \qquad \qquad -2A^\nu [F_{\mu\nu}, \nabla^\mu \nabla^\kappa A_\kappa] + 2 A^\nu [F_{\nu\kappa},\nabla_\mu \nabla^\kappa A^\mu] -  2 A^\nu [F_{\mu\kappa},\nabla_\nu \nabla^\kappa A^\mu] \bigg\}.
\end{align*}
In going to the last line, we have used relations \eqref{rel1-b} and \eqref{rel1-c} of Lemma \ref{lma:rel1}. We rewrite the last three terms using the definition of the curvature as
\begin{align*}
&-\tfrac{3}{2} A^\nu [F_{\kappa\nu},\nabla^\kappa\nabla^\mu A_\mu]
-\tfrac{3}{2} A^\nu [F_{\kappa\nu},\nabla^\mu\nabla^\kappa A_\mu]
-\half A^\nu [F_{\kappa\nu},[F^{\kappa\mu}, A_\mu]]\\
&+\tfrac{3}{2} A_\mu [ F_{\kappa\nu},\nabla^\mu \nabla^\kappa A^\nu]
+\tfrac{3}{2} A_\mu [ F_{\kappa\nu},\nabla^\kappa \nabla^\mu  A^\nu]
+\half A_\mu [ F_{\kappa\nu},[F^{\mu\kappa}, A^\nu]].
\end{align*}
The last two terms on the two lines combine and contribute to $p_4$, which becomes
$$
p^{\mu\nu}_4 = -\tfrac{3}{2} \ad F_{\kappa\lambda} \ad F^{\kappa\lambda} - \ad F^{\mu \kappa} \ad {F^\nu}_\kappa.
$$
The terms quadratic in the covariant derivatives precisely combine to give the above $p_2$.
\endproof

\begin{prop}
\label{prop:S3}
Consider the action functional 
$$ 
S_3[A] =-\tfrac{1}{3} \int \tr {F_\mu}^\nu {F_\nu}^\rho {F_\rho}^\mu = -  \tfrac{1}{6}\int \tr {F_\mu}^\nu [F_{\nu\rho} ,F^{\rho\mu}].
$$
Then the quadratic part in $A$ in a covariantly flat background $B$ is of the form $p_{2,\alpha\beta}^{\mu\nu} \nabla^\alpha\nabla^\beta + p_4^{\mu\nu}$ with
\begin{align*}
p_{2,\alpha\beta}^{\mu\nu} &= \ad F^{\mu\nu} g_{\alpha\beta}
-\half (\ad {F^\mu}_\beta g^\nu_\alpha+ \ad {F^\mu}_\alpha g^\nu_\beta )
+\half (\ad {F^\nu}_\beta g^\mu_\alpha+ \ad {F^\nu}_\alpha g^\mu_\beta )
,\\
p_4^{\mu\nu} &= \half \ad F^{\kappa\lambda} \ad F_{\kappa\lambda}~ g^{\mu\nu} 
- \ad F^{\mu\kappa} \ad {F^\nu}_\kappa,
\end{align*}
in terms of the curvature of $B$.
\end{prop}
\proof
The quadratic part in $A$ for a covariantly flat background field $B$ is 
\begin{align*}
S_3[A;B] &\equiv S_3[B+A]|_{\textrm{quad}} = -\tfrac{1}{2} \int \tr (\nabla_\mu A^\nu - \nabla^\nu A_\mu)[F_{\nu \rho},\nabla^\rho A^\mu- \nabla^\mu A^\rho]\\
&= -\half \int \tr \bigg\{ 
- A^\nu [F_{\nu\rho}, \Delta A^\rho] 
+ A_\mu [F_{\nu\rho},\nabla^\nu \nabla^\rho A^\mu] 
\\& \qquad\qquad\qquad
-A^\nu[F_{\nu\rho}, \nabla_\mu \nabla^\rho A^\mu] 
- A_\mu[F_{\nu\rho}, \nabla^\nu \nabla^\mu A^\rho]
\bigg\}
\intertext{which, expressing the last two terms more symmetrically in the derivatives, becomes}
&= - \half \int \tr \bigg\{ 
- A^\nu [F_{\nu\rho}, \Delta A^\rho] 
+ \half A_\mu [F_{\nu\rho},[F^{\nu\rho}, A^\mu] ]
\\& \qquad 
-\half A^\nu[F_{\nu\rho}, \nabla_\mu \nabla^\rho A^\mu]
-\half A^\nu[F_{\nu\rho},  \nabla^\rho\nabla_\mu A^\mu]
-\half A^\nu[F_{\nu\rho}, [{F_\mu}^\rho ,A^\mu]]\\
& \qquad
-\half A_\mu[F_{\nu\rho}, \nabla^\nu \nabla^\mu A^\rho]
-\half A_\mu[F_{\nu\rho},  \nabla^\mu\nabla^\nu A^\rho]
-\half A_\mu[F_{\nu\rho}, [F^{\nu\mu}, A^\rho]\bigg\}.
\end{align*}
The last term on the second line combines with the last term on the third line to give a contribution $-A_\mu[F^{\mu\kappa}, [F_{\nu\kappa},A^\nu]]$. Gathering all terms gives the indicated form of $p_2$ and $p_4$.
\endproof

A compatibility check between the above three propositions is based on the Bianchi identity $F_{\mu\nu;\rho} + F_{\nu\rho;\mu} + F_{\rho\mu;\nu} = 0$. Indeed, using also that $\phi_{;\mu\nu} = [F_{\nu\mu},\phi] + \phi_{;\nu\mu}$ for $f$ in the adjoint representation, we find
\begin{align*}
-\half \int \tr {F^{\mu\nu}}_{;\mu} {F_{\rho\nu}}^{;\rho} &= -\half \int \tr F^{\mu\nu} {{F_\nu}^\rho}_{;\rho\mu} \\
&=- \half \int \tr \left(F^{\mu\nu} [F_{\mu\rho} ,{F_\nu}^\rho] +   F^{\mu\nu} {{F_\nu}^\rho}_{;\mu\rho} \right)\\
&= - \half \int \tr \left({F_\nu}^\mu [{F_\mu}^\rho, {F_\rho}^\nu]  -\half F^{\mu\nu} {F_{\mu\nu;\rho}}^\rho \right).
\end{align*}
This implies that $S_1 = S_2 + 3 S_3$, which is in concordance with the above expressions for $p_2$ and $p_4$ in Propositions \ref{prop:S1}, \ref{prop:S2} and \ref{prop:S3}. 

\subsection{Higher-derivative Yang--Mills theory}
Consider the following higher-derivative Lagrangian in a background $B$:
\begin{align}
S[B+A] &= -\frac{\Lambda^2}{4} \int \tr F_{\mu\nu}F^{\mu\nu} -\frac{1}{2} \int \tr {F^{\mu\nu}}_{;\mu} {F_{\rho\nu}}^{;\rho} - \frac{\gamma}{3} \int\tr {F_\mu}^\nu {F_\nu}^\rho {F_\rho}^\mu 
\nn \\
& \quad - \frac{1}{2}\int \tr (\nabla^B_\mu  A^\mu) (\Lambda^2+\Delta_B) (\nabla^B_\nu  A^\nu)
+ \int \tr \bar C (\Lambda^2+\Delta_B) \nabla^B_\mu(\nabla_B^\mu+\ad A^\mu)C
\label{eq:babelon}
\end{align}
in which $F$ is the curvature of the connection $\nabla_B + A$. The higher-derivative terms in $S$ are precisely of the form $S_1 + \gamma S_3$; the lower-derivative terms form the gauge-fixed Yang--Mills action. In addition, there is a Faddeev--Popov term making the action $S$ invariant under the BRST-transformations \eqref{eq:BRST-bfg}. 
\begin{prop}
The quadratic part in $A$ of $S[B+A]$ in a covariantly flat background $B$ is of the form \eqref{P_B} with
\begin{align*}
p_{2,\alpha\beta}^{\mu\nu} &= (4+\gamma) \ad F^{\mu\nu} g_{\alpha\beta}
-\frac{\gamma}{2} (\ad {F^\mu}_\beta g^\nu_\alpha+ \ad {F^\mu}_\alpha g^\nu_\beta )
+\frac{\gamma}{2} (\ad {F^\nu}_\beta g^\mu_\alpha+ \ad {F^\nu}_\alpha g^\mu_\beta )
-\Lambda^2 g_{\alpha\beta} g^{\mu\nu}
\\
p_4^{\mu\nu} &= \frac{\gamma}{2} \ad F^{\kappa\lambda} \ad F_{\kappa\lambda}~ g^{\mu\nu} 
- (4+\gamma)\ad F^{\mu\kappa} \ad {F^\nu}_\kappa - \Lambda^2 \ad F^{\mu\nu} .
\end{align*}
Also, the quadratic part in $C$ of $S[B+A]$ in a covariantly flat background $B$ is of the form 
$$
-\int \tr \bar C \tilde P_B C; \qquad \tilde P_B = \Delta^2+ p_{2,\alpha\beta} \nabla^\alpha \nabla^\beta + p_4,
$$
with 
$$
p_{2,\alpha\beta}=-\Lambda^2 g_{\alpha\beta}; \qquad p_4 =0.
$$
\end{prop}
From this we can compute $a_4(P_B)$ and $a_4(\tilde P_B)$ to arrive at the counterterm for this action. Using Corollary \ref{corl:a4-4d} we obtain
\begin{align*}
a_4(P_B) &= - \frac{1}{(4\pi)^2} \frac{40+36 \gamma+3 \gamma^2}{24} \int \tr F_{\mu\nu} F^{\mu\nu} + \frac{N^2\Lambda^4}{(4 \pi)^2} \Vol(M),\\
a_4(\tilde P_B) &= \frac{1}{(4\pi)^2}  \frac{1}{12} \int \tr F_{\mu\nu} F^{\mu\nu}
+ \frac{N^2 \Lambda^4}{4 (4\pi)^2} \Vol(M).
\end{align*}

\begin{thm}
The divergent part of the one-loop effective action for the above $S$ is of the form
$$
\Gamma_{1,\infty} =\frac{1}{z} 
\frac{1}{(4\pi)^2} \frac{44+36 \gamma+3 \gamma^2}{48} \int \tr_N F_{\mu\nu} F^{\mu\nu} 
.
$$
\end{thm}
\proof
In zeta function regularization, we have 
$$
\Gamma_{1,\infty} = \frac{1}{z}\left(-\half ( a_4(P_B) - a_4(P)) + (a_4(\tilde P_B) -a_4(\tilde P))\right)
$$
which we have computed above.
\endproof
This can also be proved using dimensional regularization in the background gauge as in \cite{PS96,PS97} since $P_B$ and $\tilde P_B$ are already of the symmetric form required there.  
Note, however, that this differs from the result obtained for the one-loop effective action computed in \cite{BN80} using dimensional regularization directly on the same Lagrangian (with zero background field): there the above coefficient 44 was found to be 43. We will come back to this discrepancy in the Conclusions.

\subsection{The noncommutative Einstein--Yang--Mills system}
Consider $D_A = i \gamma^\mu \nabla_\mu$ with $\nabla_\mu = \partial_\mu+ A_\mu$ the covariant derivative lifted to the spinor bundle. Note that the gauge field $A_\mu$ and its curvature $F_{\mu\nu} = [\nabla_\mu,\nabla_\nu]$ are skew-hermitian. The Weitzenb\"ock formula gives in the case of a flat manifold:
$$
D_A^2 = - \frac{1}{2} \{ \gamma^\mu, \gamma^\nu\} \nabla_\mu \nabla_\nu - \frac{1}{2} [ \gamma^\mu, \gamma^\nu] \nabla_\mu \nabla_\nu = - \nabla_\mu \nabla^\mu - \frac{1}{2} \gamma^\mu \gamma^\nu F_{\mu\nu}  \equiv \Delta -E .
$$
The heat coefficients of $e^{-tD_A^2}$ are on a flat manifold \cite[Theorem 4.8.16]{Gil84}:
\begin{align*}
a_4(D_A^2) &= \frac{1}{(4\pi)^2} \frac{1}{360}\int \tr \left( 180 E^2+60 {E_{;\mu}}^\mu + 30 F_{\mu\nu}F^{\mu\nu}\right),\\
a_6(D_A^2) &= \frac{1}{(4\pi)^2} \frac{1}{360} \int \tr \bigg\{ 8 F_{\mu \nu;\kappa} F^{\mu\nu;\kappa} + 2 {F_{\mu\nu}}^{;\mu} {F^{\rho\nu}}_{;\rho} + 12 F_{\mu\nu} {F^{\mu\nu;\rho}}_\rho - 12 {F_\mu}^\nu {F_\nu}^\rho {F_\rho}^\mu\\ 
& \qquad \qquad \qquad \qquad \qquad
+ 6 {{{E_{;\mu}}^\mu}_\nu}^\nu 
+ 60 
E {E_{;\mu}}^\mu + 30 E_{;\mu}E^{;\mu} 
+ 60 E^3 + 30 E F_{\mu\nu}F^{\mu\nu}\bigg\}.
\end{align*}
We can rearrange these expressions by partial integration:
\begin{align*}
a_4(D_A^2) &= \frac{1}{(4\pi)^2} \frac{1}{360}\int \tr \left( 180 E^2+ 30 F_{\mu\nu}F^{\mu\nu}\right),\\
a_6(D_A^2) &= \frac{1}{(4\pi)^2} \frac{1}{360} \int \tr \bigg\{ 2 {F_{\mu\nu}}^{;\mu} {F_{\rho\nu}}^{;\rho} + 4 F_{\mu\nu} {F^{\mu\nu;\rho}}_\rho - 12 {F_\mu}^\nu {F_\nu}^\rho {F_\rho}^\mu
+ 30 E {E_{;\mu}}^\mu + 60 E^3 . 
\bigg\}
\end{align*}

\begin{prop}
\begin{align*}
a_4(D_A^2) &= \frac{1}{8 \pi^2}  \left( -\frac{1}{3} \int \tr_N F_{\mu\nu}F^{\mu\nu}  \right),\\
a_6(D_A^2) &= \frac{1}{8 \pi^2} \left( \frac{2}{15} \int\tr_N {F^{\mu\nu}}_{;\mu} {F^{\rho\nu}}_{;\rho} +  \frac{23}{45} \int \tr_N {F_{\mu}}^\nu{F_{\nu}}^\rho {F_{\rho}}^\mu  \right).
\end{align*}
\end{prop}
\proof
We compute the two terms in the above expression for $a_4(D_A^2)$, using $\tr \gamma^\mu \gamma^\nu \gamma^\rho \gamma^\sigma = 4 \left( g^{\mu\nu} g^{\rho\sigma} - g^{\mu\rho}g^{\nu \sigma} + g^{\mu \sigma} g^{\nu \rho}\right)$:
\begin{align*}
\tr E^2 &= \frac{1}{4} \tr \gamma^\mu \gamma^\nu \gamma^\rho \gamma^\sigma F_{\mu\nu} F_{\rho\sigma} = -2 \tr_N F_{\mu\nu}F^{\mu\nu},\\
\tr F_{\mu\nu} &= 4 \tr_N F_{\mu\nu}F^{\mu\nu}.
\end{align*}
The coefficient in front of $\int\tr_N F_{\mu\nu}F^{\mu\nu}$ in $a_4$ thus becomes
$$
\frac{1}{(4 \pi)^2} \frac{180\times (-2) + 30 \times 4}{360} = -\frac{2}{3 (4 \pi)^2} .
$$
For $a_6$ we start with the simplest expressions:
\begin{multline*}
\frac{1}{(4\pi)^2} \frac{1}{360} \int \tr \bigg\{ 2 {F_{\mu\nu}}^{;\mu} {F^{\rho\nu}}_{;\rho} + 4 F_{\mu\nu} {F^{\mu\nu;\rho}}_\rho - 12 {F_\mu}^\nu {F_\nu}^\rho {F_\rho}^\mu \bigg\} 
\\ =
\frac{1}{8\pi^2} \int \tr_N \bigg\{ \frac{1}{90} {F_{\mu\nu}}^{;\mu} {F^{\rho\nu}}_{;\rho} + \frac{1}{45} F_{\mu\nu} {F^{\mu\nu;\rho}}_\rho - \frac{1}{15} {F_\mu}^\nu {F_\nu}^\rho {F_\rho}^\mu \bigg\}.
\end{multline*}
For the term of second order in $E$, we find with the above rule for the trace of four gamma matrices:
$$
\tr E {E_{;\kappa}}^\kappa = \frac{1}{4} \tr \gamma^\mu \gamma^\nu \gamma^\rho \gamma^\sigma F_{\mu\nu} {F_{\rho\sigma;\kappa}}^{\kappa} = \tr_N \left( F_{\mu\nu} {F^{\nu\mu;\kappa}}_{\kappa}  - F_{\mu\nu} {F^{\mu\nu;\kappa}}_{\kappa}  \right) = - 2 \tr_N  F_{\mu\nu} {F^{\mu\nu;\kappa}}_{\kappa}  .
$$
Including the coefficients, this contributes to $a_6$ with
$$
\frac{1}{(4\pi)^2} \frac{1}{360} \int \tr 30 E {E_{;\mu}}^\mu   = \frac{1}{8 \pi^2} \left(-\frac{1}{12} \tr_N F_{\mu\nu} {F^{\mu\nu;\kappa}}_{\kappa}   \right).
$$
Finally, we compute
\begin{align*}
\tr E^3 &=\frac{1}{8} \tr \gamma^{\mu_1}\gamma^{\mu_2}\gamma^{\mu_3}\gamma^{\mu_4}\gamma^{\mu_5}\gamma^{\mu_6} F_{\mu_1\mu_2}F_{\mu_3\mu_4}F_{\mu_5\mu_6} \\
&= \frac{1}{8} \bigg( 
\tr \gamma^{\mu_2}\gamma^{\mu_3}\gamma^{\mu_4}\gamma^{\mu_5} F_{\mu\mu_2}F_{\mu_3\mu_4}{F_{\mu_5}}^\mu
-\tr \gamma^{\mu_1}\gamma^{\mu_3}\gamma^{\mu_4}\gamma^{\mu_5} F_{\mu_1\mu}F_{\mu_3\mu_4}{F_{\mu_5}}^\mu\\
&\quad
+\tr \gamma^{\mu_1}\gamma^{\mu_2}\gamma^{\mu_4}\gamma^{\mu_5} F_{\mu_1\mu_2}F_{\mu\mu_4}{F_{\mu_5}}^\mu
-\tr \gamma^{\mu_1}\gamma^{\mu_2}\gamma^{\mu_3}\gamma^{\mu_5} F_{\mu_1\mu_2}F_{\mu_3\mu}{F_{\mu_5}}^\mu\bigg)\\
&= 4 {F_\mu}^\nu {F_\nu}^\rho {F_\rho}^\mu
\end{align*}
applying once more the above rule for the trace of four gamma matrices in going to the last line. Thus, this contributes to $a_6$ with
$$ 
\frac{1}{(4\pi)^2} \frac{1}{360} \int \tr (60 E^3) = \frac{1}{8 \pi^2} \frac{1}{3}\int \tr  {F_\mu}^\nu {F_\nu}^\rho {F_\rho}^\mu.
$$
Gathering these expressions we find
$$
a_6 (D_A^2) = \frac{1}{8 \pi^2}  \bigg\{ \frac{1}{90} {F_{\mu\nu}}^{;\mu} {F^{\rho\nu}}_{;\rho} - \frac{11}{180} F_{\mu\nu} {F^{\mu\nu;\rho}}_\rho + \frac{4}{15} {F_\mu}^\nu {F_\nu}^\rho {F_\rho}^\mu \bigg\}.
$$
This can be reduced to the displayed form by the relation
$$
F_{\mu\nu} {F^{\mu\nu;\rho}}_\rho = - 2 {F_{\mu\nu}}^{;\mu} {F^{\rho\nu}}_{;\rho} - 4 {F_\mu}^\nu {F_\nu}^\rho {F_\rho}^\mu
$$
which can be easily established using the Bianchi identity.
\endproof

\begin{thm}
The action $S^\Lambda[A] = f_0 a_4(x,D_A^2) + \Lambda^{-2} f_{-2} a_6 (x,D_A^2)$ equals
\begin{multline*}
S^\Lambda[A] 
= -\frac{4 f_{-2} \Lambda^{-2}}{15(8 \pi^2)} \bigg\{
\frac{-5 f_0 \Lambda^2}{f_{-2}} \left(- \frac{1}{4}\int \tr_N F_{\mu\nu}F^{\mu\nu} \right)
-\frac{1}{2} \int\tr_N {F^{\mu\nu}}_{;\mu} {F^{\rho\nu}}_{;\rho}
\\ 
+\frac{23}{4}\left(- \frac{1}{3} \int \tr_N {F_{\mu}}^\nu{F_{\nu}}^\rho {F_{\rho}}^\mu \right)
\bigg\}.
\end{multline*}
\end{thm}

\begin{rem}
Let us check this result by comparing the free part with the coefficients $c_k$ computed in Theorem \ref{thm:free-sa}. The proof of Proposition \ref{prop:S1} yields for the free part of $S[A]$ ({\it i.e.} background field $B=0$):
$$
\frac{4f_0}{3 (8 \pi^2)} \left( - \frac{1}{2} \int \tr_N A^\mu ( \Delta g_{\mu\nu} +\partial_\mu \partial_\nu) A^\nu \right)
-\frac{4 f_{-2} \Lambda^{-2}}{15(8 \pi^2)} \left(-\frac{1}{2}\int \tr_N A^\mu (\Delta^2 g_{\mu\nu} +\Delta \partial_\mu \partial_\nu) A^\nu\right)
$$
which is in concordance with $c_0 = 1/24 \pi^2, c_1=1/120 \pi^2$ in 
$$
- c_0 f_0  \int \tr_N \hat F^{\mu\nu} \hat F_{\mu\nu} + c_1 f_{-2} \Lambda^{-2} \int \tr_N \hat F^{\mu\nu} \Delta \hat F_{\mu\nu}
$$
as appearing in Theorem \ref{thm:free-sa}.
\end{rem}

Thus, the action $f_0 a_4 + \Lambda^{-2} f_{-2} a_6$ is of the form of the action appearing on the first line of Equation \eqref{eq:babelon} from which we immediately conclude
\begin{corl}
The divergent part of the one-loop effective action for $S^\Lambda[A]$ (with $f_{2-m}=0$ for all $m>2$) is
$$
\Gamma_{1,\infty} =\frac{1}{z} 
\frac{1}{(4\pi)^2} \frac{5603}{768} \int \tr_N F_{\mu\nu} F^{\mu\nu}. 
$$
\end{corl}
Note that this counterterm could easily be subtracted from the spectral action and absorbed through a redefinition of $f_0$
. However, as before the resulting coefficient $5603/512$ for the coupling constant appears to have little to do with the usual counterterm (and the $\beta$ function) for Yang--Mills theory. It is time to comment on these discrepancies.

\section{Conclusions}

We have established renormalizability of the asymptotically expanded spectral action for the Yang--Mills system on a flat background manifold. By naive power counting we found that this higher-derivative field theory is superrenormalizable. The only divergent Feynman graphs appear at one loop and give rise to a gauge invariant counterterm. We have computed the form of this counterterm using heat invariants in a background gauge field and zeta function regularization. The counterterm can be absorbed in the spectral action by a redefinition of $f(0)$. 
This gives rise to a non-perturbative $\beta$ function for the coupling $f(0)$.

Let us now comment on the explicit computations done in the previous section for the expansion of the spectral action truncated at 6'th order, {\it i.e.} for the action $f_0 a_4 + \Lambda^{-2}f_{-2} a_6$. It appears that the counterterms are not in agreement with physics, or even with other results in the literature for the same higher-derivative gauge theory \cite{BN80}.
Such differences between regularization schemes were reported in \cite{MRR95a} also for higher-order derivative theories with a Lagrangian of the form 
$$
\tr F_{\mu\nu} F^{\mu\nu} + \Lambda^{-4}  \tr F_{\mu\nu} \Delta^2 F^{\mu\nu},
$$
using higher-derivative Pauli-Villars regularization and comparing with the usual one-loop $\beta$ function (no higher-derivatives) for Yang--Mills theory.
Since from a physical point of view it is absolutely crucial that the renormalized quantities (eg. $\beta$ functions) do not depend on the renormalization scheme that is exploited, this is a no-go result. However, for the case considered in {\it loc.~cit.~} this discrepancy has been resolved in \cite{MRR95b} using dimensional regularization in combination with higher-derivative regulators, giving the correct $\beta$ function for Yang--Mills theory at one loop. 

These last results give hope for an explicit computation of the one-loop effective action for the asymptotically expanded spectral action, using dimensional regularization directly on the Lagrangian, rather than zeta functions to regularize the functional determinant in a background field gauge. This would allow for a comparison between the renormalized asymptotically expanded spectral action and the renormalization of Yang--Mills theory. 
Such a computation is part of future work and will appear elsewhere, since it lies outside of the scope of the present paper, whose aim was to rigorously establish renormalizability of the asymptotically expanded spectral action as claimed in \cite{Sui11b}. 

\section*{Acknowledgements}
The author would like to thank Olivier Babelon, Dirk Kreimer and Matilde Marcolli for useful correspondence and discussions. We gratefully acknowledge the anonymous referee for helpful comments. 
Caltech is acknowledged for hospitality and financial support during a visit in April 2011. The ESF is thanked for financial support under the program `Interactions of Low-Dimensional Topology and Geometry with Mathematical Physics (ITGP)'. NWO is acknowledged for support under VENI-project 639.031.827.

\appendix

\section{Heat expansion for higher order Laplacians}
We recall a result from \cite{Gil80}, which is crucial in the above. Gilkey studied Laplacians of higher order on a vector bundle on a Riemannian manifold $(M,g)$. Given a connection $\nabla$ on this vector bundle, the Laplacian is $\Delta = -g^{\mu\nu} \nabla_\mu \nabla_\nu$. Generalizing this to higher orders, Gilkey considered differential operators on a vector bundle  of the form
\begin{equation}
\label{eq:higher-laplacian}
P = \Delta^k  + p_{2,\alpha\beta} \nabla_{\alpha}\nabla_{\beta} \Delta^{k-2}
+ (-1)^k \left( p_{3,\alpha_1 \cdots \alpha_{2k-3}} \nabla_{\alpha_1} \cdots \nabla_{\alpha_{2k-3}} + \cdots + p_{2k}  \right)
\end{equation}
with $k \geq 2$ and $p_{2,\alpha\beta} = p_{2,\beta \alpha}$ is symmetric. The remaining endomorphisms $p_{l}$ ($l>2$) of the vector bundle need not be symmetric.

\begin{thm}[Gilkey \cite{Gil80}]
\label{thm:gilkey}
The heat kernel of an operator of the form \eqref{eq:higher-laplacian} satisfies
 asymptotically (as $t \to 0$):
$$
\tr_{L^2} e^{-t P} \sim \sum_{n \geq 0} t^{(m-n)/{2k}} a_n(P).
$$
The first three coefficients are given on a 
$m$-dimensional Riemannian manifold by
\begin{align*}
a_0(P) &= \frac{1}{k} (4\pi)^{-m/2} \frac{\Gamma(m/2k)}{\Gamma(m/2)}\int_M \tr(I)  \dvol(x), \\
a_2(P) &= \frac{1}{k} (4\pi)^{-m/2} \frac{ \Gamma((m-2)/2k) }{\Gamma((m-2)/2)}
\int_M \tr \left(-\frac{1}{6} R_{\mu\nu\mu\nu} I + \frac{1}{m k}p_{2,\mu\nu}g^{\mu\nu} \right) \dvol(x),\\
a_4(P) &= \frac{1}{k} (4\pi)^{-m/2} \frac{ \Gamma((m-4)/2k) }{\Gamma((m-4)/2)}
\int_M \tr \frac{1}{360(m-2)(m+2)} \\
& \quad \times \bigg(
5(m^2-4) R^{\mu\nu}_{\phantom{\mu\nu}\mu\nu}R^{\rho\sigma}_{\phantom{\rho\sigma}\rho\sigma} - 2(m^2-4) R^\mu_{\phantom\mu\nu\mu\rho} R_\sigma^{\phantom\sigma\nu\sigma\rho}
+30(m^2-4) F_{\mu\nu}F^{\mu\nu}\\
& \quad - \frac{60(m+2)}{k} R^{\mu\nu}_{\phantom{\mu\nu}\mu\nu}p_{2,\rho\sigma}g^{\rho\sigma} + \frac{120(m+2)}{k} R_{\mu}^{\phantom\mu\nu\mu\sigma}p_{2,\nu\sigma}+  \frac{180m + 360(k-2)}{mk^2} {p_{2,\mu}}^\mu {p_{2,\nu}}^\nu\\
&\quad  + \frac{360m + 720(k-2)}{mk^2} {p_{2,\mu}}^\nu {p_{2,\nu}}^\mu
-\frac{720 (m+2)}{k S(\delta^{k-2})} S(p_4)
\bigg) \dvol(x)
\end{align*}
where $S$ is totally symmetric contraction and $\delta$ is the Kronecker delta.  \end{thm}

\begin{corl}
\label{corl:a4-4d}
Let $M$ be a 4-dimensional flat Riemannian manifold ($R_{\mu\nu\rho\sigma}=0$) and suppose that $P_B$ is an operator of the form \eqref{eq:higher-laplacian} with 
\begin{align*}
p_{2,\alpha\beta}^{\mu\nu} &= \a F^{\mu\nu} g_{\alpha\beta}
+ \b ( {F^\mu}_\beta g^\nu_\alpha+  {F^\mu}_\alpha g^\nu_\beta )
+\c( {F^\nu}_\beta g^\mu_\alpha+ {F^\nu}_\alpha g^\mu_\beta ).
\end{align*}
Then
\begin{multline*}
a_4(P) = \frac{1}{(4 \pi)^2}\left( \frac{4}{12}  + \frac{1}{2k} ((\a+\b)(-\a+\c)+\b \c) \right)  \int \tr_N  F_{\mu\nu} F^{\mu\nu} \\- \frac{1}{(4 \pi)^2} \frac{1}{k S(\delta^{k-2})} \int \tr S(p_4).
\end{multline*}
\end{corl}
\proof
First, the quotient of the two gamma functions gives in the limit $m\to 4$ a factor of $k$. The coefficient of $S(p_4)$ is then easily found. The terms involving $p_2$ are
$$
\frac{1}{(4 \pi)^2}\frac{1}{48 k}\int \tr \left( {p_{2,\alpha}}^\alpha {p_{2,\beta}}^\beta + 2 {p_{2,\alpha}}^\beta {p_{2,\beta}}^\alpha  \right)
$$
From the above form of $p_2$ one readily finds
$
{p^{\mu\nu}_{2,\alpha}}^\alpha = (4 \a +2\b -2\c) F^{\mu\nu}
$
and by symmetry in $\alpha$ and $\beta$ we also find
$$
\tr {p_{2,\alpha}}^\beta {p_{2,\beta}}^\alpha = \a p_{2,\alpha\beta}^{\mu\nu} g^{\alpha\beta} F_{\nu\mu} 
+ 2 \b p_{2,\alpha\beta}^{\mu\nu} g^\beta_\mu {F_\nu}^\alpha 
+ 2 \c p_{2,\alpha\beta}^{\mu\nu} g^\beta_\nu {F_\mu}^\alpha.
$$
For each term we compute
\begin{align*}
p_{2,\alpha\beta}^{\mu\nu} g^{\alpha\beta} &= (4 \a + 2\b -2 \c) F^{\mu\nu} ,\\
p_{2,\alpha\beta}^{\mu\nu} g^\beta_\mu &= (\a-\b-5\c) {F_\alpha }^\nu ,\\
p_{2,\alpha\beta}^{\mu\nu} g^\beta_\nu &= (-\a-5\b-\c) {F_\alpha }^\mu.
\end{align*}
This combines to give 
$$
\tr \left( {p_{2,\alpha}}^\alpha {p_{2,\beta}}^\beta + 2 {p_{2,\alpha}}^\beta {p_{2,\beta}}^\alpha \right)
= 24(-\a^2 -\a\b + \a\c+2\b\c) \tr F_{\mu\nu}F^{\mu\nu}
$$
which is of the desired form.
\endproof

\begin{corl}
\label{corl:a4-4d-k3}
Let $M$ be a 4-dimensional flat Riemannian manifold ($R_{\mu\nu\rho\sigma}=0$) and suppose that $P_B$ is an operator of the form \eqref{eq:higher-laplacian} with $k=3$ and
\begin{align*}
p_{2,\alpha\beta}^{\mu\nu} &= \a F^{\mu\nu} g_{\alpha\beta}
+ \b ( {F^\mu}_\beta g^\nu_\alpha+  {F^\mu}_\alpha g^\nu_\beta )
+\c( {F^\nu}_\beta g^\mu_\alpha+ {F^\nu}_\alpha g^\mu_\beta ),\\
-p_{4,\alpha\beta}^{\mu\nu} &= \d \ad F_{\kappa\lambda} \ad F^{\kappa\lambda} g_{\alpha\beta} g^{\mu\nu}
+ \e \ad {F_\kappa}^\mu F^{\kappa\nu} g_{\alpha \beta} 
+\f \ad F_{\kappa\lambda} \ad F^{\kappa\lambda} g^\mu_\alpha g^\nu_\beta  
+\g \ad {F_\alpha}^\lambda \ad F_{\beta\lambda} g^{\mu\nu}\\
& \quad
+ \h \ad {F_\alpha}^\mu \ad {F_\beta}^\nu
+\k \ad F^{\mu\kappa} \ad F_{\beta\kappa} g^\nu_\alpha
 +\l \ad F^{\nu\kappa} \ad F_{\beta\kappa} g^\mu_\alpha 
 +\m g^{\mu\nu} g_{\alpha\beta}.
\end{align*}
Then
\begin{multline*}
a_4(P) = \frac{ N \m}{(4 \pi)^2}\frac{4}{3}\Vol(M) \\ + \frac{1}{(4 \pi)^2} \frac{1}{12} \bigg( 2 ((\a+\b)(-\a+\c)+\b \c)  + 16 \d+ 4\e + 4\f +4\g +\h +\k +\l\bigg) \int \tr_N  F_{\mu\nu} F^{\mu\nu} .
\end{multline*}
\end{corl}
\proof
In addition to the previous Corollary, one computes that (modulo commutators)
\begin{align*}
-S(p_4^{\mu\nu}) &= 
(4 \d +\f+\g) \ad F_{\kappa\lambda} \ad F^{\kappa\lambda} g^{\mu\nu}
+( 4 \e +\h +\k +\l)\ad {F_\kappa}^\mu F^{\kappa\nu} 
 +4 \m g^{\mu\nu}.
\end{align*}
Then, contracting the indices $\mu$ and $\nu$ gives
$$
- \tr S(p_4) = (4 (4 \d +\f+\g) + 4 \e +\h +\k +\l)  \tr_N \ad F_{\kappa\lambda} \ad F^{\kappa\lambda}
 +16 \m .
$$
\endproof


\begin{thebibliography}{99}

\bibitem{Avr91}
I.~Avramidi.
\newblock A covariant technique for calculation of one loop effective action.
\newblock {\em Nucl.Phys.} B355 (1991)  712--754.

\bibitem{BN80}
O.~Babelon and M.~Namazie.
\newblock Comment on the ghost problem in a higher derivative {Y}ang-{M}ills
  theory.
\newblock {\em J.Phys.A} A13 (1980)  L27--L30.

\bibitem{BGV92}
N.~Berline, E.~Getzler, and M.~Vergne.
\newblock {\em Heat Kernels and Dirac Operators}.
\newblock Springer-Verlag, Berlin, 1992.

\bibitem{BDK90}
F.~Brandt, N.~Dragon, and M.~Kreuzer.
\newblock Lie algebra cohomology.
\newblock {\em Nucl. Phys.} B332 (1990)  250.

\bibitem{BGO90}
T.~P. Branson, P.~B. Gilkey, and B.~{\O}rsted.
\newblock Leading terms in the heat invariants.
\newblock {\em Proc. Amer. Math. Soc.} 109 (1990)  437--450.

\bibitem{CC96}
A.~H. Chamseddine and A.~Connes.
\newblock Universal formula for noncommutative geometry actions: {U}nifications
  of gravity and the standard model.
\newblock {\em Phys. Rev. Lett.} 77 (1996)  4868--4871.

\bibitem{CC97}
A.~H. Chamseddine and A.~Connes.
\newblock The spectral action principle.
\newblock {\em Commun. Math. Phys.} 186 (1997)  731--750.

\bibitem{CCM07}
A.~H. Chamseddine, A.~Connes, and M.~Marcolli.
\newblock {Gravity and the standard model with neutrino mixing}.
\newblock {\em Adv. Theor. Math. Phys.} 11 (2007)  991--1089.

\bibitem{C94}
A.~Connes.
\newblock {\em Noncommutative Geometry}.
\newblock Academic Press, San Diego, 1994.

\bibitem{CM07}
A.~Connes and M.~Marcolli.
\newblock {\em Noncommutative Geometry, Quantum Fields and Motives}.
\newblock AMS, Providence, 2008.

\bibitem{CC06}
A.~Connes and A.~H. Chamseddine.
\newblock {Inner fluctuations of the spectral action}.
\newblock {\em J. Geom. Phys.} 57 (2006)  1--21.

\bibitem{Dix91}
J.~A. Dixon.
\newblock {Calculation of BRS cohomology with spectral sequences}.
\newblock {\em Commun. Math. Phys.} 139 (1991)  495--526.

\bibitem{DTV85}
M.~Dubois-Violette, M.~Talon, and C.~M. Viallet.
\newblock {BRS} algebras: Analysis of the consistency equations in gauge
  theory.
\newblock {\em Commun. Math. Phys.} 102 (1985)  105.

\bibitem{DTV85b}
M.~Dubois-Violette, M.~Talon, and C.~M. Viallet.
\newblock Results on {BRS} cohomology in gauge theory.
\newblock {\em Phys. Lett.} B158 (1985)  231.

\bibitem{DHTV91}
M.~Dubois-Violette, M.~Henneaux, M.~Talon, and C.-M. Viallet.
\newblock {Some results on local cohomologies in field theory}.
\newblock {\em Phys. Lett.} B267 (1991)  81--87.

\bibitem{FS80}
L.~Faddeev and A.~Slavnov.
\newblock {\em Gauge Fields. Introduction to Quantum Theory}.
\newblock Benjaming Cummings, 1980.

\bibitem{Gil80}
P.~B. Gilkey.
\newblock The spectral geometry of the higher order {L}aplacian.
\newblock {\em Duke Math. J.} 47 (1980)  511--528.

\bibitem{Gil84}
P.~B. Gilkey.
\newblock {\em Invariance theory, the heat equation, and the {A}tiyah-{S}inger
  index theorem}, volume~11 of {\em Mathematics Lecture Series}.
\newblock Publish or Perish Inc., Wilmington, DE, 1984.

\bibitem{GVF01}
J.~M. Gracia-Bond{\'\i}a, J.~C. V\'arilly, and H.~Figueroa.
\newblock {\em Elements of {N}oncommutative {G}eometry}.
\newblock Birkh\"auser, Boston, 2001.

\bibitem{Haw77}
S.~W. Hawking.
\newblock {Zeta Function Regularization of Path Integrals in Curved
  Space-Time}.
\newblock {\em Commun. Math. Phys.} 55 (1977)  133.

\bibitem{ILV11}
B.~Iochum, C.~Levy, and D.~Vassilevich.
\newblock Spectral action beyond the weak-field approximation.
\newblock arXiv:1108.3749.

\bibitem{MRR95a}
C.~Martin and F.~Ruiz~Ruiz.
\newblock {Higher covariant derivative Pauli-Villars regularization does not
  lead to a consistent QCD}.
\newblock {\em Nucl.Phys.} B436 (1995)  545--581.

\bibitem{MRR95b}
C.~Martin and F.~Ruiz~Ruiz.
\newblock {Higher covariant derivative regulators and nonmultiplicative
  renormalization}.
\newblock {\em Phys.Lett.} B343 (1995)  218--224.

\bibitem{PS96}
P.~I. Pronin and K.~V. Stepanyantz.
\newblock {One-loop effective action for an arbitrary theory}.
\newblock {\em Teor. Mat. Fyz.} 109 (1996)  215--231.

\bibitem{PS97}
P.~I. Pronin and K.~V. Stepanyantz.
\newblock {One-loop counterterms for higher derivative regularized
  Lagrangians}.
\newblock {\em Phys. Lett.} B414 (1997)  117--122.

\bibitem{Sla71}
A.~A. Slavnov.
\newblock {Invariant regularization of nonlinear chiral theories}.
\newblock {\em Nucl. Phys.} B31 (1971)  301--315.

\bibitem{Sla72b}
A.~A. Slavnov.
\newblock {Invariant regularization of gauge theories}.
\newblock {\em Teor. Mat. Fiz.} 13 (1972)  174--177.

\bibitem{Sui11}
W.~D. {\noopsort{suijlekom}}van Suijlekom.
\newblock Perturbations and operator trace functions.
\newblock {\em J. Funct. Anal.} 260 (2011)  2483--2496.

\bibitem{Sui11b}
W.~D. {\noopsort{suijlekom}}van Suijlekom.
\newblock {Renormalization of the spectral action for the Yang-Mills system}.
\newblock {\em JHEP} 1103 (2011)  146.

\bibitem{Vas03}
D.~V. Vassilevich.
\newblock {Heat kernel expansion: User's manual}.
\newblock {\em Phys. Rept.} 388 (2003)  279--360.

\end{thebibliography}
\newcommand{\noopsort}[1]{}\def\cprime{$'$}

\end{document}